\documentclass[preprint]{sig-alternate-05-2015}

\usepackage{enumerate}
\usepackage[linkcolor=black,colorlinks=true,citecolor=blue,urlcolor=blue]{hyperref} 

\begin{document}

\toappear{ISSAC'16, July 19--22, 2016, Waterloo, ON, Canada. \\
  ACM ISBN.
  DOI: http://dx.doi.org/10.1145/2930889.2930936}

\CopyrightYear{2016}
\setcopyright{acmlicensed}
\conferenceinfo{ISSAC'16,}{July 19--22, 2016, Waterloo, ON, Canada}
\isbn{978-1-4503-4380-0/16/07}\acmPrice{\$15.00}
\doi{http://dx.doi.org/10.1145/2930889.2930936}

\title{Fast Computation of Shifted Popov Forms of Polynomial Matrices via
Systems of Modular Polynomial Equations}

\numberofauthors{1}
\author{
\alignauthor
Vincent Neiger \\
\affaddr{ENS de Lyon \\
Laboratoire LIP, CNRS, Inria, UCBL, U. Lyon}
\email{vincent.neiger@ens-lyon.fr}
}

\renewcommand\leq\leqslant
\newcommand{\storeArg}{} 
\newcounter{notationCounter}
\newcommand{\bigO}[1]{\mathcal{O}(#1)} 
\newcommand{\bigOPar}[1]{\mathcal{O}\left(#1\right)} 
\newcommand{\softO}[1]{\widetilde{\mathcal{O}}(#1)} 
\newcommand{\polmultime}[1]{\mathsf{M}(#1)}
\newcommand{\expmatmul}{\omega} 
\newcommand{\algoname}[1]{{\normalfont\textsc{#1}}}
\renewcommand{\ge}{\geqslant} 
\renewcommand{\le}{\leqslant} 
\newcommand{\bzero}{\mathbf{0}} 
\newcommand{\Zint}[1]{\{0,\ldots,#1-1\}} 
\newcommand{\ZZ}{\mathbb{Z}} 
\newcommand{\NN}{\mathbb{N}} 
\newcommand{\ZZp}{\mathbb{Z}_{> 0}} 
\newcommand{\NNp}{\mathbb{N}_{> 0}} 
\newcommand{\var}{X} 
\newcommand{\field}{\mathbb{K}} 
\newcommand{\vecSpace}{\mathfrak{E}} 
\newcommand{\polSpace}{\mathfrak{F}} 
\newcommand{\polRing}{\field[\var]} 
\newcommand{\matSpace}[1][\rdim]{\renewcommand\storeArg{#1}\matSpaceAux} 
\newcommand{\polMatSpace}[1][\rdim]{\renewcommand\storeArg{#1}\polMatSpaceAux} 
\newcommand{\matSpaceAux}[1][\storeArg]{\field^{\storeArg \times #1}} 
\newcommand{\polMatSpaceAux}[1][\storeArg]{\polRing^{\storeArg \times #1}} 
\newcommand{\row}[1]{\mathbf{\MakeLowercase{#1}}} 
\newcommand{\col}[1]{\mathbf{\MakeLowercase{#1}}} 
\newcommand{\mat}[1]{\mathbf{\MakeUppercase{#1}}} 
\newcommand{\matCoeff}[1]{\MakeLowercase{#1}} 
\newcommand{\vecc}[1]{\mathbf{#1}} 
\newcommand{\trsp}[1]{#1^\mathsf{T}} 
\newcommand{\rdim}{m} 
\newcommand{\cdim}{n} 
\newcommand{\maxDeg}{\delta}
\newcommand{\spol}[1][1]{\setcounter{notationCounter}{#1}\addtocounter{notationCounter}{15} \Alph{notationCounter}} 
\newcommand{\smat}[1][1]{\setcounter{notationCounter}{#1}\mat{\Alph{notationCounter}}} 
\newcommand{\spolmat}[1][1]{\setcounter{notationCounter}{#1}\mathbf{\Alph{notationCounter}}} 
\newcommand{\matrow}[2]{{#1}_{#2,*}} 
\newcommand{\matcol}[2]{{#1}_{*,#2}} 
\newcommand{\diag}[1]{\,\mathrm{diag}(#1)}  
\newcommand{\idMat}[1][\rdim]{\mat{I}_{#1}} 
\newcommand{\leadingMat}[2][\unishift]{\mathrm{lm}_{#1}(#2)} 
\newcommand{\shift}[2][s]{#1_{#2}} 
\newcommand{\shifts}[1][s]{\mathbf{#1}} 
\newcommand{\sshifts}[1][\shifts]{|#1|} 
\newcommand{\shiftSpace}[1][\rdim]{\ZZ^{#1}} 
\newcommand{\unishift}{\mathbf{0}} 
\newcommand{\rdeg}[2][]{\mathrm{rdeg}_{{#1}}(#2)} 
\newcommand{\cdeg}[2][]{\mathrm{cdeg}_{{#1}}(#2)} 
\newcommand{\order}{\sigma} 
\newcommand{\evMat}{\mat{E}} 
\newcommand{\mul}{\cdot} 
\newcommand{\mulmat}{\mat{M}} 
\newcommand{\rowvec}[1][1]{\setcounter{notationCounter}{#1}\addtocounter{notationCounter}{21} \row{\alph{notationCounter}}} 
\newcommand{\minDeg}{\delta}
\newcommand{\minDegs}{\boldsymbol{\delta}}
\newcommand{\sumVec}[1]{|#1|} 
\newtheorem{pbm}{Problem}
\newtheorem{dfn}{Definition}[section]
\newtheorem{thm}[dfn]{Theorem}
\newtheorem{cor}[dfn]{Corollary}
\newtheorem{prop}[dfn]{Proposition}
\newtheorem{lem}[dfn]{Lemma}
\newtheorem{algo}{Algorithm}

\newcommand{\expandMat}{\mathcal{E}}
\newcommand{\quoExp}{\alpha}
\newcommand{\remExp}{\beta}
\newcommand{\expand}[1]{\widetilde{#1}}
\newcommand{\degExp}{ \delta }

\newcommand{\modulus}[1][m]{\mathfrak{#1}}
\newcommand{\Modulus}{\mathfrak{M}}
\newcommand{\modSpace}{\polRing_{\neq 0}}
\newcommand{\degMod}{\order}
\newcommand{\nbeq}{\cdim} 
\newcommand{\nbun}{\rdim} 
\newcommand{\sys}{\mat{F}} 
\newcommand{\sysSpace}{\polMatSpace[\nbun][\nbeq]} 
\newcommand{\sol}{\row{p}} 
\newcommand{\solSpace}{\polMatSpace[1][\nbun]} 
\newcommand{\msb}{\mat{P}} 
\newcommand{\msbSpace}{\polMatSpace[\nbun]} 
\newcommand{\popov}{\mat{P}} 
\newcommand{\popovSpace}{\polMatSpace[\rdim]} 
\newcommand{\smith}{\mat{S}} 
\newcommand{\piv}{\pi} 
\newcommand{\degLin}{\delta}
\newcommand{\degLins}{\boldsymbol{\delta}}
\newcommand{\degDet}{\sigma(\mat{A})} 
\newcommand{\any}{\ast}
\newcommand{\anyMat}{\boldsymbol{\ast}}
\newcommand{\colParLin}[2]{\mathcal{L}^{\mathrm{c}}_{#1}(#2)}
\newcommand{\rowParLin}[2]{\mathcal{L}^{\mathrm{r}}_{#1}(#2)}
\newcommand{\pivdeg}[2][]{\pi\mathrm{deg}_{{#1}}(#2)} 
\newcommand{\subVec}[3]{#1_{[#2:#3]}}  
\newcommand{\subMat}[5]{#1_{[#2:#3,#4:#5]}}  
\newcommand{\amp}[1][\shifts]{\mathrm{amp}(#1)}
\newcommand{\perm}[1][\shifts]{\boldsymbol{\pi}^{#1}}

	\maketitle

\begin{abstract}
  We give a Las Vegas algorithm which computes the shifted Popov form of an
  $\rdim \times \rdim$ nonsingular polynomial matrix of degree $d$ in expected
  $\softO{\rdim^\expmatmul d}$ field operations, where $\expmatmul$ is the
  exponent of matrix multiplication and $\softO{\cdot}$ indicates that
  logarithmic factors are omitted. This is the first algorithm in
  $\softO{\rdim^\expmatmul d}$ for shifted row reduction with arbitrary shifts.

  Using partial linearization, we reduce the problem to the case $d \le \lceil
  \sigma/\rdim \rceil$ where $\sigma$ is the generic determinant bound, with
  $\sigma / \rdim$ bounded from above by both the average row degree and the
  average column degree of the matrix. The cost above becomes
  $\softO{\rdim^\expmatmul \lceil \sigma/\rdim \rceil}$, improving upon the
  cost of the fastest previously known algorithm for row reduction, which is
  deterministic. 

  Our algorithm first builds a system of modular equations whose solution set
  is the row space of the input matrix, and then finds the basis in shifted
  Popov form of this set. We give a deterministic algorithm for this second
  step supporting arbitrary moduli in $\softO{\nbun^{\expmatmul-1} \degMod}$
  field operations, where $\nbun$ is the number of unknowns and $\degMod$ is
  the sum of the degrees of the moduli.  This extends previous results with the
  same cost bound in the specific cases of order basis computation and M-Pad\'e
  approximation, in which the moduli are products of known linear factors.
\end{abstract}

%
%

\vspace{-0.15cm}

\section{Introduction}
\label{sec:intro}

In this paper, we consider two problems of linear algebra over the ring
$\polRing$ of univariate polynomials, for some field $\field$: computing the
shifted Popov form of a matrix, and solving systems of modular equations.

\subsection{Shifted Popov form}
\label{subsec:popov}

A polynomial matrix $\mat{P}$ is row reduced~\cite[Section~6.3.2]{Kailath80} if
its rows have some type of minimal degree (we give precise definitions below).
Besides, if $\mat{P}$ satisfies an additional normalization property, then it
is said to be in Popov form~\cite[Section~6.7.2]{Kailath80}. Given a matrix
$\mat{A}$, the efficient computation of a (row) reduced form of $\mat{A}$ and
of the Popov form of $\mat{A}$ has received a lot of attention
recently~\cite{GiJeVi03,SarSto11,GuSaStVa12}.

In many applications one rather considers the degrees of the rows of $\mat{P}$
shifted by some integers which specify degree weights on the columns of
$\mat{P}$, for example in list-decoding
algorithms~\cite{Alekhnovich05,Busse08}, robust Private Information
Retrieval~\cite{DeGoHe12}, and more generally in polynomial versions of the
Coppersmith method~\cite{CohHen12,CohHen15}. A well-known specific shifted
Popov form is the Hermite form; there has been recent progress on its fast
computation~\cite{GupSto11,Gupta11,ZhoLab16}. The case of an arbitrary shift
has been studied in~\cite{BeLaVi06}.

For a \emph{shift} $\shifts = (\shift{1},\ldots,\shift{\cdim}) \in
\shiftSpace[\cdim]$, the \emph{$\shifts$-degree} of $\row{p} =
[p_1,\ldots,p_\cdim] \in \polMatSpace[1][\cdim]$ is $\max_{1\le j\le \cdim}
(\deg(p_j) + \shift{j})$; the \emph{$\shifts$-row degree} of $\mat{P} \in
\polMatSpace[\rdim][\cdim]$ is $\rdeg[\shifts]{\mat{P}} = (d_1,\ldots,d_\rdim)$
with $d_i$ the $\shifts$-degree of the $i$-th row of $\mat{P}$. Then, the
\emph{$\shifts$-leading matrix} of $\mat{P} = [p_{i,j}]_{ij}$ is the matrix
$\leadingMat[\shifts]{\mat{P}} \in \matSpace[\rdim][\cdim]$ whose entry $(i,j)$
is the coefficient of degree $d_i - \shift{j}$ of $p_{i,j}$.

Now, we assume that $\rdim \le \cdim$ and $\mat{P}$ has full rank. Then,
$\mat{P}$ is said to be \emph{$\shifts$-reduced}~\cite{Kailath80,BeLaVi06} if
$\leadingMat[\shifts]{\mat{P}}$ has full rank. For a full rank $\mat{A} \in
\polMatSpace[\rdim][\cdim]$, an \emph{$\shifts$-reduced form of $\mat{A}$} is
an $\shifts$-reduced matrix $\mat{P}$ whose row space is the same as that of
$\mat{A}$; by row space we mean the $\polRing$-module generated by the rows of
the matrix. Equivalently, $\mat{P}$ is left-unimodularly equivalent to
$\mat{A}$ and the tuple $\rdeg[\shifts]{\mat{P}}$ sorted in nondecreasing order
is lexicographically minimal among the $\shifts$-row degrees of all matrices
left-unimodularly equivalent to $\mat{A}$.

Specific $\shifts$-reduced matrices are those in \emph{$\shifts$-Popov
form}~\cite{Kailath80,BecLab00,BeLaVi06}, as defined below. One interesting
property is that the $\shifts$-Popov form is canonical: there is a unique
$\shifts$-reduced form of $\mat{A}$ which is in $\shifts$-Popov form, called
the \emph{$\shifts$-Popov form of~$\mat{A}$}.

\vspace{-0.1cm}
\begin{dfn}[Pivot]
\label{dfn:pivot}
Let $\row{p} = [p_j]_j \in \polMatSpace[1][\cdim]$ be nonzero and let $\shifts
\in \shiftSpace[\cdim]$. The \emph{$\shifts$-pivot index} of $\row{p}$ is the
largest index $j$ such that $\rdeg[\shifts]{\row{p}} = \deg(p_j) + \shift{j}$.
Then we call $p_j$ and $\deg(p_j)$ the \emph{$\shifts$-pivot entry} and the
\emph{$\shifts$-pivot degree} of $\row{p}$.
\end{dfn}
\vspace{-0.05cm}

We remark that adding a constant to the entries of $\shifts$ does not change
the notion of $\shifts$-pivot. For example, we will sometimes assume
$\min(\shifts)=0$ without loss of generality.

\vspace{-0.05cm}
\begin{dfn}[Shifted Popov form]
\label{dfn:popov}
Let $\rdim \le \cdim$, let $\mat{P} \in \polMatSpace[\rdim][\cdim]$ be
full rank, and let $\shifts \in \shiftSpace[\cdim]$. Then, $\mat{P}$ is said to
be in \emph{$\shifts$-Popov form} if the $\shifts$-pivot indices of its rows
are strictly increasing, the corresponding $\shifts$-pivot entries are monic,
and in each column of $\mat{P}$ which contains a pivot the nonpivot entries
have degree less than the pivot entry.

In this case, the \emph{$\shifts$-pivot degree} of $\mat{P}$ is $\minDegs =
(\minDeg_1,\ldots,\minDeg_\rdim) \in \NN^\rdim$, with $\minDeg_i$ the
$\shifts$-pivot degree of the $i$-th row of $\mat{P}$.
\end{dfn}

Here, although we will encounter Popov forms of rectangular matrices in
intermediate nullspace computations, our main focus is on computing shifted
Popov forms of \emph{square nonsingular matrices}. For the general case,
studied in~\cite{BeLaVi06}, a fast solution would require further developments.
A square matrix in $\shifts$-Popov form has its $\shifts$-pivot entries on the
diagonal, and its $\shifts$-pivot degree is the tuple of degrees of its
diagonal entries and coincides with its column degree.

\begin{center}
\fbox{ \begin{minipage}{7.8cm}
\begin{pbm}[Shifted Popov normal form]
  \label{pbm:popov}

\begin{tabular}{p{1cm}p{5.5cm}}
\emph{Input:} &
   the base field $\field$,
   a nonsingular matrix $\mat{A} \in \polMatSpace[\rdim]$,
   a shift $\shifts\in\shiftSpace$. \\
\emph{Output:} &
  the $\shifts$-Popov form of $\mat{A}$.
\end{tabular}
\end{pbm}
\end{minipage}
}
\end{center}

Two well-known specific cases are the Popov form~\cite{Popov72,Kailath80} for
the \emph{uniform} shift $\shifts=\unishift$, and the Hermite
form~\cite{Hermite1851,Kailath80} for the shift $\shifts[h] = (0,\delta, 2
\delta,\ldots,(\rdim-1)\delta) \in \NN^\rdim$ with $\delta=\rdim
\deg(\mat{A})$~\cite[Lemma~2.6]{BeLaVi06}. For a broader perspective on shifted
reduced forms, we refer the reader to~\cite{BeLaVi06}.

\smallskip
For such problems involving $\rdim \times \rdim$ matrices of degree $d$, one
often wishes to obtain a cost bound similar to that of polynomial matrix
multiplication in the same dimensions: $\softO{\rdim^\expmatmul d}$ operations
in $\field$. Here, $\expmatmul$ is so that we can multiply $\rdim \times \rdim$
matrices over a commutative ring in $\bigO{\rdim^\expmatmul}$ operations in
that ring, the best known bound being $\expmatmul <
2.38$~\cite{CopWin90,LeGall14}. For example, one can compute
$\unishift$-reduced~\cite{GiJeVi03,GuSaStVa12},
$\unishift$-Popov~\cite{SarSto11}, and Hermite~\cite{Gupta11,ZhoLab16} forms of
$\rdim \times \rdim$ nonsingular matrices of degree $d$ in
$\softO{\rdim^\expmatmul d}$ field operations.

Nevertheless, $d$ may be significantly larger than the average degree of the
entries of the matrix, in which case the cost $\softO{\rdim^\expmatmul d}$
seems unsatisfactory. Recently, for the computation of order
bases~\cite{Storjohann06,ZhoLab12}, nullspace bases~\cite{ZhLaSt12},
interpolation bases~\cite{JeNeScVi15,JeNeScVi16}, and matrix
inversion~\cite{ZhLaSt15}, fast algorithms do take into account some types of
average degrees of the matrices rather than their degree. Here, in particular,
we achieve a similar improvement for the computation of shifted Popov forms of
a matrix.

Given $\mat{A} = [a_{i,j}]_{ij} \in \polMatSpace[\rdim]$, we denote by
$\degDet$ the \emph{generic bound for
  $\deg(\det(\mat{A}))$}~\cite[Section~6]{GuSaStVa12}, that is,
\begin{equation}
  \label{eqn:degDet}
  \degDet = \max_{\pi \in S_\rdim} \sum_{1 \le i\le \rdim} \overline{\deg}(a_{i,\pi_i})
\end{equation}
where $S_\rdim$ is the set of permutations of $\{1,\ldots,\rdim\}$, and
$\overline{\deg}(p)$ is defined over $\polRing$ as $\overline{\deg}(0) = 0$ and
$\overline{\deg}(p) = \deg(p)$ for $p\neq 0$. We have $\deg(\det(\mat{A})) \le
\degDet \le \rdim \deg(\mat{A})$, and $\degDet \le
\min(\sumVec{\rdeg{\mat{A}}}, \sumVec{\cdeg{\mat{A}}})$ with
$\sumVec{\rdeg{\mat{A}}}$ and $\sumVec{\cdeg{\mat{A}}}$ the sums of the row and
column degrees of $\mat{A}$. We note that $\degDet$ can be substantially
smaller than $\sumVec{\rdeg{\mat{A}}}$ and $\sumVec{\cdeg{\mat{A}}}$, for
example if $\mat{A}$ has one row and one column of uniformly large degree and
other entries of low degree.

\begin{thm}
  \label{thm:popov}
  There is a Las Vegas randomized algorithm which solves
  Problem~\ref{pbm:popov} in expected $\softO{\rdim^\expmatmul
  \lceil\degDet/\rdim\rceil} \linebreak[0] \subseteq \softO{\rdim^\expmatmul
    \deg(\mat{A})}$ field operations.
\end{thm}

The ceiling function indicates that the cost is $\softO{\rdim^\expmatmul}$ when
$\degDet$ is small compared to $\rdim$, in which case $\mat{A}$ has mostly
constant entries. Here we are mainly interested in the case $\rdim \in
\bigO{\degDet}$: the cost bound may be written $\softO{\rdim^{\expmatmul-1}
\degDet}$ and is both in $\softO{\rdim^{\expmatmul-1} \sumVec{\rdeg{\mat{A}}}}$
and $\softO{\rdim^{\expmatmul-1} \sumVec{\cdeg{\mat{A}}}}$.

\begin{table}[h]
  \scriptsize
  \centering
  \begin{tabular}[h]{cp{3.1cm}cc}
    Ref.                 &  Problem             &  Cost bound                                                &                \\ \hline \\[-0.25cm]
    \cite{HafMcCur91}    & Hermite form         &  $\softO{\rdim^4 d}$                                       &                \\
    \cite{StoLab96}      & Hermite form         &  $\softO{\rdim^{\expmatmul+1} d}$                          &                \\
    \cite{Villard96}     & Popov \& Hermite forms &  $\softO{\rdim^{\expmatmul+1} d + (\rdim d)^\expmatmul}$ &                \\
    \cite{Alekhnovich02,Alekhnovich05} 
                         & weak Popov form & $\softO{\rdim^{\expmatmul+1} d}$                                &                \\
    \cite{MulSto03}      & Popov \& Hermite forms & $\bigO{\rdim^3 d^2}$                                     &                \\
    \cite{GiJeVi03}      & $\unishift$-reduction & $\softO{\rdim^\expmatmul d}$                              & $\star$        \\
    \cite{SarSto11}      & Popov form of $\unishift$-reduced & $\softO{\rdim^\expmatmul d}$                  &                \\
    \cite{GupSto11}      & Hermite form & $\softO{\rdim^{\expmatmul} d}$                                     & $\star$        \\
    \cite{GuSaStVa12}    & $\unishift$-reduction & $\softO{\rdim^{\expmatmul} d}$                            &                \\
    \cite{ZhoLab16}      & Hermite form & $\softO{\rdim^{\expmatmul} d}$                                     &                \\
    \cite{GuSaStVa12}$+$\cite{SarSto11}
                         & $\shifts$-Popov form for any $\shifts$ & $\softO{\rdim^{\expmatmul} (d+\mu)}$     &                \\
    \emph{Here}          & $\shifts$-Popov form for any $\shifts$ &
                                                  $\softO{\rdim^{\expmatmul} \lceil \degDet/\rdim \rceil}$   & $\star$
  \end{tabular}
  \caption{Fast algorithms for shifted reduction problems \textmd{($d =
    \deg(\mat{A})$; $\star = $ probabilistic; $\mu = \max(\shifts)-\min(\shifts)$)}.}
    \label{tbl:cost_popov}
\vspace{-0.3cm}
\end{table}

Previous work on fast algorithms related to Problem~\ref{pbm:popov} is
summarized in Table~\ref{tbl:cost_popov}. The fastest known algorithm for the
$\unishift$-Popov form is deterministic and has cost $\softO{\rdim^\expmatmul
d}$ with $d = \deg(\mat{A})$; it first computes a $\unishift$-reduced form of
$\mat{A}$~\cite{GuSaStVa12}, and then its $\unishift$-Popov form via
normalization~\cite{SarSto11}. Obtaining the Hermite form in
$\softO{\rdim^\expmatmul d}$ was first achieved by a probabilistic algorithm
in~\cite{Gupta11}, and then deterministically in~\cite{ZhoLab16}.

For an arbitrary $\shifts$, the algorithm in~\cite{BeLaVi06} is fraction-free
and uses a number of operations that is, depending on $\shifts$, at least
quintic in $\rdim$ and quadratic in $\deg(\mat{A})$.

When $\shifts$ is not uniform there is a folklore solution based on the fact
that $\mat{Q}$ is in $\shifts$-Popov form if and only if $\mat{Q} \mat{D}$ is
in $\unishift$-Popov form, with $\mat{D} =
\diag{X^{\shift{1}},\ldots,X^{\shift{\rdim}}}$ and assuming $\shifts \ge 0$.
Then, this solution computes the $\unishift$-Popov form $\mat{P}$ of $\mat{A}
\mat{D}$ using~\cite{GuSaStVa12,SarSto11} and returns $\mat{P} \mat{D}^{-1}$.
This approach uses $\softO{\rdim^\expmatmul (d + \mu)}$ operations where $\mu =
\max(\shifts) - \min(\shifts)$, which is not satisfactory when $\mu$ is large.
For example, its cost for computing the Hermite form is
$\softO{\rdim^{\expmatmul+2} d}$. This is the worst case since one can assume
without loss of generality that $\mu \in \bigO{\rdim \deg(\det(\mat{A}))}
\subseteq \bigO{\rdim^2 d}$~\cite[Appendix A]{JeNeScVi16}.

Here we obtain, to the best of our knowledge, the best known cost bound
$\softO{\rdim^\expmatmul \lceil\degDet/\rdim\rceil} \subseteq
\softO{\rdim^{\expmatmul} d}$ for an arbitrary shift $\shifts$. This removes
the dependency in $\mu$, which means in some cases a speedup by a factor
$\rdim^2$. Besides, this is also an improvement for both specific cases
$\shifts=\unishift$ and $\shifts=\shifts[h]$ when $\mat{A}$ has unbalanced
degrees.

\smallskip
One of the main difficulties in row reduction algorithms is to control the size
of the manipulated matrices, that is, the number of coefficients from $\field$
needed for their dense representation. A major issue when dealing with
arbitrary shifts is that the size of an $\shifts$-reduced form of $\mat{A}$ may
be beyond our target cost. This is a further motivation for focusing on the
computation of the $\shifts$-Popov form of $\mat{A}$: by definition, the sum of
its column degrees is $\deg(\det(\mat{A}))$, and therefore its size is at most
$\rdim^2 + \rdim \deg(\det(\mat{A}))$, independently of $\shifts$.

Consider for example $\mat{A} = \left[\begin{smallmatrix} \mat{B}_1 & \mat{0}
  \\ \mat{0} & \mat{B}_2 \end{smallmatrix}\right]$ for any $\unishift$-reduced
$\mat{B}_1$ and $\mat{B}_2$ in $\polMatSpace[\rdim]$. Then, taking $\shifts =
(0,\ldots,0,d,\ldots,d)$ with $d > 0$, $\left[\begin{smallmatrix} \mat{B}_1 &
    \mat{0} \\ \mat{C} & \mat{B}_2
\end{smallmatrix}\right]$ is an $\shifts$-reduced form of $\mat{A}$ for any
$\mat{C} \in \polMatSpace[\rdim]$ with $\deg(\mat{C}) \le d$; for some
$\mat{C}$ it has size $\Theta(\rdim^2 d)$, with $d$ arbitrary large
independently of $\deg(\mat{A})$.

Furthermore, the size of the unimodular transformation leading from $\mat{A}$
to $\popov$ may be beyond the target cost, which is why fast algorithms for
$\unishift$-reduction and Hermite form do not directly perform unimodular
transformations on $\mat{A}$ to reduce the degrees of its entries. Instead,
they proceed in two steps: first, they work on $\mat{A}$ to find some equations
which describe its row space, and then they find a basis of solutions to these
equations in $\unishift$-reduced form or Hermite form. We will follow a similar
two-step strategy for an arbitrary shift.

It seems that some new ingredient is needed, since for both $\shifts=\unishift$
and $\shifts=\shifts[h]$ the fastest algorithms use shift-specific properties
at some point of the process: namely, the facts that a $\unishift$-reduced form
of $\mat{A}$ has degree at most $\deg(\mat{A})$ and that the Hermite form of
$\mat{A}$ is triangular.

As in~\cite{GupSto11}, we first compute the Smith form $\smith$ of $\mat{A}$
and partial information on a right unimodular transformation $\mat{V}$; this is
where the probabilistic aspect comes from. This gives a description of the row
space of $\mat{A}$ as the set of row vectors $\row{p} \in
\polMatSpace[1][\rdim]$ such that $\row{p} \mat{V} = \row{q} \smith$ for some
$\row{q}\in \polMatSpace[1][\rdim]$. Since $\smith$ is diagonal, this can be
seen as a system of modular equations: the second step is the fast computation
of a basis of solutions in $\shifts$-Popov form, which is our new ingredient.

\subsection{Systems of modular equations}
\label{subsec:modsys}

Hereafter, $\modSpace$ denotes the set of nonzero polynomials. We fix some
moduli $\Modulus = (\modulus_1,\ldots,\modulus_\nbeq) \in \modSpace^\nbeq$, and
for $\mat{A}, \mat{B} \in \polMatSpace[\nbun][\nbeq]$ we write $\mat{A} =
\mat{B} \bmod \Modulus$ if there exists $\mat{Q} \in
\polMatSpace[\rdim][\cdim]$ such that $\mat{A} = \mat{B} + \mat{Q}
\diag{\Modulus}$. Given $\sys \in \sysSpace$ specifying the equations, we call
\emph{solution for $(\Modulus,\sys)$} any $\sol \in \solSpace$ such that $\sol
\sys = 0 \bmod \Modulus$.

The set of all such $\row{p}$ is a $\polRing$-submodule of $\solSpace$ which
contains $\mathrm{lcm}(\modulus_1,\ldots,\modulus_\nbeq) \solSpace$, and is
thus free of rank $\nbun$ \cite[p.~146]{Lang02}. Then, we represent any basis
of this module as the rows of a matrix $\msb \in \msbSpace$, called a
\emph{solution basis for $(\Modulus,\sys)$}.
Here, for example for the application to Problem~\ref{pbm:popov}, we are
interested in such bases that are $\shifts$-reduced, in which case $\msb$ is
said to be an \emph{$\shifts$-minimal solution basis for $(\Modulus,\sys)$}.
The unique such basis which is in $\shifts$-Popov form is called the
\emph{$\shifts$-Popov solution basis for $(\Modulus,\sys)$}.

\begin{figure}[h]
  \centering
\fbox{ \begin{minipage}{7.8cm}
\begin{pbm}[Minimal solution basis]
  \label{pbm:modsys}

\begin{tabular}{p{0.9cm}p{6.4cm}}
  \emph{Input:} &
    the base field $\field$,
    moduli $\Modulus = (\modulus_1,\ldots,\modulus_\nbeq) \in \modSpace^\nbeq$,
    a matrix $\sys \in \sysSpace$ such that $\deg(\matcol{\sys}{j}) < \deg(\modulus_j)$,
    a shift $\shifts\in\shiftSpace$. \\
  \emph{Output:} &
  an $\shifts$-minimal solution basis for $(\Modulus,\sys)$.
\end{tabular}
\end{pbm}
\end{minipage}
}
\end{figure}

Well-known specific cases of this problem are \emph{Hermite-Pad\'e
approximation} with a single equation modulo some power of $X$, and
\emph{M-Pad\'e approximation}~\cite{Beckermann92,BarBul92} with moduli that are
products of known linear factors. Moreover, an \emph{$\shifts$-order basis for
$\sys$ and $(\order_1,\ldots,\order_\nbeq)$}~\cite{ZhoLab12} is an
$\shifts$-minimal solution basis for $(\Modulus,\sys)$ with $\Modulus =
(X^{\order_1},\ldots,X^{\order_\nbeq})$.

An overview of fast algorithms for Problem~\ref{pbm:modsys} is given in
Table~\ref{tbl:cost_msb}. For M-Pad\'e approximation, and thus in particular
for order basis computation, there is an algorithm to compute the
$\shifts$-Popov solution basis using $\softO{\nbun^{\expmatmul-1} \degMod}$
operations, with $\degMod = \deg(\modulus_1) + \cdots +
\deg(\modulus_\nbeq)$~\cite{JeNeScVi16}.  Here, for $\nbeq \in \bigO{\nbun}$,
we extend this result to arbitrary moduli.

\begin{thm}
  \label{thm:modsys}
Assuming $\nbeq \in \bigO{\nbun}$, there is a deterministic algorithm which
solves Problem~\ref{pbm:modsys} using $\softO{\nbun^{\expmatmul-1} \degMod}$
field operations, with $\order = \deg(\modulus_1) + \cdots +
\deg(\modulus_\nbeq)$, and returns the $\shifts$-Popov solution basis for
$(\Modulus,\sys)$.
\end{thm}

We note that Problem~\ref{pbm:modsys} is a minimal interpolation basis
problem~\cite{BecLab00,JeNeScVi15} when the so-called \emph{multiplication
matrix} $\mulmat$ is block diagonal with companion blocks. Indeed, $\row{p}$ is
a solution for $(\Modulus,\sys)$ if and only if $\row{p}$ is an
\emph{interpolant for $(\evMat,\mulmat)$}~\cite[Definition~1.1]{JeNeScVi15},
where $\evMat \in \matSpace[\nbun][\degMod]$ is the concatenation of the
coefficient vectors of the columns of $\sys$ and $\mulmat \in
\matSpace[\degMod]$ is $\diag{\mulmat_1,\ldots,\mulmat_\nbeq}$ with $\mulmat_j$
the companion matrix associated with $\modulus_j$. In this context, the
multiplication $\row{p} \mul \evMat$ defined by $\mulmat$ as
in~\cite{BecLab00,JeNeScVi15} precisely corresponds to $\row{p} \sys \bmod
\Modulus$.

In particular, Theorem~\ref{thm:modsys} follows
from~\cite[Theorem~1.4]{JeNeScVi15} when $\degMod \in \bigO{\nbun}$. If some of
the moduli have small degree, we use this result for base cases of our
recursive algorithm.

\vspace{-0.15cm}
\begin{table}[h]
  { \scriptsize
  \centering
  \begin{tabular}[h]{ccp{1.45cm}p{2.6cm}}
    Ref.      &  Cost bound                         &  Moduli & Particularities \\ \hline \\[-0.25cm]
    \cite{Beckermann92,BarBul92}     &  $\bigO{\nbun^2 \degMod^2}$  &  split  & \\
    \cite{BecLab94}     &  $\bigO{\nbun \degMod^2}$  &  $\modulus_j = X^{\degMod/\nbeq}$  & partial basis \\
    \cite{BecLab94}     &  $\softO{\nbun^\expmatmul \degMod}$  & $\modulus_j = X^{\degMod/\nbeq}$ & \\
    \cite{GiJeVi03}     &  $\softO{\nbun^\expmatmul \degMod/\nbeq}$  & $\modulus_j = X^{\degMod/\nbeq}$ & \\
    \cite{Storjohann06} &  $\softO{\nbun^{\expmatmul} \lceil \degMod/\nbun\rceil}$  & $\modulus_j = X^{\degMod/\nbeq}$ & partial basis, $\sshifts \le \degMod$ \\
    \cite{ZhoLab12}     & $\softO{\nbun^{\expmatmul} \lceil\degMod/\nbun\rceil}$  & $\modulus_j = X^{\degMod/\nbeq}$ & $\sshifts \le \degMod$\\
    \cite{CJNSV15}      & \vtop{\hbox{\strut $\softO{\nbun^{\expmatmul-1} \degMod}$,}\hbox{\strut probabilistic}} & any   & returns a single small degree solution \\
    \cite{JeNeScVi15}   & $\softO{\nbun^{\expmatmul-1} \degMod}$  & split & $\sshifts \le \degMod$ \\
    \cite{JeNeScVi15}   & $\softO{\nbun \degMod^{\expmatmul-1}}$  & any & $\shifts$-Popov, $\degMod \in \bigO{\nbun}$ \\
    \cite{JeNeScVi16}   & $\softO{\nbun^{\expmatmul-1} \degMod}$  & split & $\shifts$-Popov \\
    \emph{Here}   & $\softO{\nbun^{\expmatmul-1} \degMod}$  & any & $\shifts$-Popov
  \end{tabular}
  \caption{Fast algorithms for Problem~\ref{pbm:modsys} \textmd{($\nbeq \in
    \bigO{\nbun}$; \emph{partial basis} $=$ returns small degree rows of an
    $\shifts$-minimal solution basis; \emph{split} $=$ product of known linear
  factors)}.} \label{tbl:cost_msb}
  }
\end{table}
\vspace{-0.15cm}

In the case of M-Pad\'e approximation, knowing the moduli as products of linear
factors leads to rewriting the problem as a minimal interpolation basis
computation with $\mulmat$ in Jordan form~\cite{BecLab00,JeNeScVi15}. Since
$\mulmat$ is upper triangular, one can then rely on recurrence relations to
solve the problem iteratively~\cite{Beckermann92,BarBul92,BecLab94,BecLab00}.
The fast algorithms in~\cite{BecLab94,GiJeVi03,ZhoLab12,JeNeScVi15,JeNeScVi16},
beyond the techniques used to achieve efficiency, are essentially
divide-and-conquer versions of this iterative solution and are thus based on the
same recurrence relations.

However, for arbitrary moduli the matrix $\mulmat$ is not triangular and there
is no such recurrence in general. Then, a natural idea is to relate solution
bases to nullspace bases: Problem~\ref{pbm:modsys} asks to find $\msb$ such
that there is some quotient $\mat{Q}$ with $[\msb | \mat{Q}] \mat{N} = \mat{0}$
for $\mat{N} = \trsp{[\trsp{\sys} | -\diag{\Modulus}]}$.  More precisely,
$[\msb | \mat{Q}]$ can be obtained as a $\shifts[u]$-minimal nullspace basis of
$\mat{N}$ for the shift $\shifts[u] = (\shifts-\min(\shifts),\unishift) \in
\NN^{\nbun+\nbeq}$.

Using recent ingredients from~\cite{GupSto11,JeNeScVi16} outlined in the next
paragraphs, the main remaining difficulty is to deal with this nullspace
problem when $\nbeq = 1$. Here, we give a $\softO{\nbun^{\expmatmul-1}
\degMod}$ algorithm to solve it using its specific properties: $\mat{N}$ is the
column $\trsp{[\trsp{\sys}|\modulus_1]}$ with $\deg(\sys) < \deg(\modulus_1) =
\degMod$, and the last entry of $\shifts[u]$ is $\min(\shifts[u])$. First, when
$\max(\shifts[u]) \in \bigO{\degMod}$ we show that $[\msb | \mat{Q}]$ can be
efficiently obtained as a submatrix of the $\shifts[u]$-Popov order basis for
$\mat{N}$ and order $\bigO{\degMod}$. Then, when $\max(\shifts[u])$ is large
compared to $\degMod$ and assuming $\shifts[u]$ is sorted non-decreasingly,
$\msb$ has a lower block triangular shape. We show how this shape can be
revealed, along with the $\shifts$-pivot degree of $\msb$, using a
divide-and-conquer approach which splits $\shifts[u]$ into two shifts of
amplitude about $\max(\shifts[u])/2$.

Then, for $\nbeq\ge 1$ we use a divide-and-conquer approach on $\nbeq$ which is
classical in such contexts: two solution bases $\msb^{(1)}$ and $\msb^{(2)}$
are computed recursively in shifted Popov form and are multiplied together to
obtain the $\shifts$-minimal solution basis $\msb^{(2)} \msb^{(1)}$ for
$(\Modulus,\sys)$. However this product is usually not in $\shifts$-Popov form
and may have size beyond our target cost. Thus, as in~\cite{JeNeScVi16},
instead of computing $\msb^{(2)} \msb^{(1)}$, we use $\msb^{(2)}$ and
$\msb^{(1)}$ to deduce the $\shifts$-pivot degree of $\msb$.

In both recursions above, we focus on finding the $\shifts$-pivot degree of
$\msb$. Using ideas and results from~\cite{GupSto11,JeNeScVi16}, we show that
this knowledge about the degrees in $\msb$ allows us to complete the
computation of $\msb$ within the target cost.

\section{Fast computation of the shifted Popov solution basis}
\label{sec:modsys}

Hereafter, we call \emph{$\shifts$-minimal degree of $(\Modulus,\sys)$} the
$\shifts$-pivot degree $\minDegs$ of the $\shifts$-Popov solution basis for
$(\Modulus,\sys)$; $\minDegs$ coincides with the column degree of this basis. A
central result for the cost analysis is that $\sumVec{\minDegs} = \minDeg_1 +
\cdots + \minDeg_\rdim$ is at most $\degMod = \deg(\modulus_1) + \cdots +
\deg(\modulus_\nbeq)$. This is classical for M-Pad\'e
approximation~\cite[Theorem~4.1]{BarBul92} and holds for minimal interpolation
bases in general (see for example \cite[Lemma~7.17]{JeNeScVi15}).

\subsection{Solution bases from nullspace bases and fast algorithm for known
minimal degree}
\label{subsec:mnb_solbas_knownmindeg}

This subsection summarizes and slightly extends results from~\cite[Section
3]{GupSto11}. We first show that the $\shifts$-Popov solution basis for
$(\Modulus,\sys)$ is the principal $\nbun \times \nbun$ submatrix of the
$\shifts[u]$-Popov nullspace basis of $\trsp{ [ \trsp{\sys} | \diag{\Modulus} ]
}$ for some $\shifts[u] \in \shiftSpace[\nbun+\nbeq]$.

\vspace{-0.1cm}
\begin{lem}
  \label{lem:mnb_solbas}
  Let $\Modulus = (\modulus_1,\ldots,\modulus_\nbeq) \in \modSpace^\nbeq$,
  $\shifts\in\shiftSpace$, $\sys \in \sysSpace$ with $\deg(\matcol{\sys}{j}) <
  \deg(\modulus_j)$, $\msb \in \msbSpace$, and $\shifts[w] \in
  \shiftSpace[\nbeq]$ be such that $\max(\shifts[w]) \le \min(\shifts)$.  Then,
  $\msb$ is the $\shifts$-Popov solution basis for $(\Modulus,\sys)$ if and
  only if $[ \msb | \mat{Q} ]$ is the $\shifts[u]$-Popov nullspace basis of
  $\trsp{ [ \trsp{\sys} | \diag{\Modulus} ] }$ for some $\mat{Q} \in
  \polMatSpace[\nbun][\nbeq]$ and $\shifts[u] = (\shifts,\shifts[w]) \in
  \shiftSpace[\nbun+\nbeq]$. In this case, $\deg(\mat{Q}) < \deg(\msb)$ and $[
  \msb | \mat{Q} ]$ has $\shifts$-pivot index $(1,2,\ldots,\nbun)$.
\end{lem}
\begin{proof}
Let $\mat{N} = \trsp{[ \trsp{\sys} | \diag{\Modulus} ]}$. It is easily verified
that $\msb$ is a solution basis for $(\Modulus,\sys)$ if and only if there is
some $\mat{Q}\in \polMatSpace[\nbun][\nbeq]$ such that $[\msb | \mat{Q}]$ is a
nullspace basis of $\mat{N}$.

Now, having $\deg(\matcol{\sys}{j}) < \deg(\modulus_j)$ implies that any
$[\row{p} | \row{q}] \in \polMatSpace[1][(\nbun+\nbeq)]$ in the nullspace of
$\mat{N}$ satisfies $\deg(\row{q}) < \deg(\row{p})$, and since
$\max(\shifts[w]) \le \min(\shifts)$ we get $\rdeg[{\shifts[w]}]{\row{q}} <
\rdeg[\shifts]{\row{p}}$. In particular, for any matrix $[\mat{P} | \mat{Q}]
\in \polMatSpace[\nbun][(\nbun+\nbeq)]$ such that $[\mat{P} | \mat{Q}] \mat{N}
= 0$, we have $\leadingMat[{\shifts[u]}]{[\mat{P}|\mat{Q}]} =
[\leadingMat[\shifts]{\mat{P}} | \mat{0}]$. This implies that $\mat{P}$ is in
$\shifts$-Popov form if and only if $[\mat{P} | \mat{Q}]$ is in
$\shifts[u]$-Popov form with $\shifts$-pivot index $(1,\ldots,\nbun)$.
\vspace{-0.1cm}
\end{proof}

We now show that, when we have \emph{a priori} knowledge about the
$\shifts$-pivot entries of a $\shifts$-Popov nullspace basis, it can be
computed efficiently via an $\shifts$-Popov order basis.

\vspace{-0.1cm}
\begin{lem}
  \label{lem:mnb_known_mindeg}
  Let $\shifts \in \shiftSpace[\rdim+\cdim]$ and let $\mat{N} \in
  \polMatSpace[(\rdim+\cdim)][\cdim]$ be of full rank. Let $\mat{B} \in
  \polMatSpace[\rdim][(\rdim+\cdim)]$ be the $\shifts$-Popov nullspace basis
  for $\mat{N}$, $(\piv_1,\ldots,\piv_\rdim)$ be its $\shifts$-pivot index,
  $(\minDeg_1,\ldots,\minDeg_\rdim)$ be its $\shifts$-pivot degree, and
  $\maxDeg \ge \deg(\mat{B})$ be a degree bound. Then, let $\shifts[u] =
  (\shift[u]{1}, \ldots, \shift[u]{\rdim+\cdim}) \in \ZZ_{\le 0}^{\rdim+\cdim}$
  with 
  \vspace{-0.2cm}
  \[ \shift[u]{j} = \begin{cases} 
      -\maxDeg -1 & \;\; \text{if } j \not\in \{\piv_1,\ldots,\piv_\rdim\}, \\ 
      -\minDeg_i & \;\; \text{if } j = \piv_i.
    \end{cases}
  \vspace{-0.2cm}
  \]
  Writing $(\order_1,\ldots,\order_\cdim)$ for the column degree of $\mat{N}$,
  let $\tau_j = \order_j + \maxDeg +1$ for $1 \le j \le \cdim$ and let
  $\mat{A}$ be the $\shifts[u]$-Popov order basis for $\mat{N}$ and
  $(\tau_1,\ldots,\tau_\cdim)$. Then, $\mat{B}$ is the submatrix of $\mat{A}$
  formed by its rows at indices $\{\piv_1,\ldots,\piv_\rdim\}$.
\end{lem}
\begin{proof}
  First, $\mat{B}$ is in $\shifts[u]$-Popov form with
  ${\rdeg[{\shifts[u]}]{\mat{B}} = \unishift}$. Define $\mat{C} \in
  \polMatSpace[(\rdim+\cdim)]$ whose $i$-th row is $\matrow{\mat{B}}{j}$ if $i =
  \piv_j$ and $\matrow{\mat{A}}{i}$ if $i \not\in
  \{\piv_1,\ldots,\piv_\rdim\}$:
  we want to prove $\mat{C} = \mat{A}$.
  
  Let $\row{p} = [p_j]_j \in \polMatSpace[1][(\rdim+\cdim)]$ be a row of
  $\mat{A}$, and assume $\rdeg[{\shifts[u]}]{\row{p}} < 0$.  This means
  $\deg(p_j) < -\shift[u]{j}$ for all $j$, so that $\deg(\row{p}) <
  \max(-\shifts[u]) = \maxDeg+1$. Then, for all $1 \le j\le \cdim$ we have
  $\deg( \row{p} \matcol{\mat{N}}{j} ) < \order_j + \maxDeg +1= \tau_j$, and
  from $\row{p} \matcol{\mat{N}}{j} = 0 \bmod X^{\tau_j}$ we obtain $\row{p}
  \matcol{\mat{N}}{j} = 0$, which is absurd by minimality of $\mat{B}$.  As a
  result, $\rdeg[{\shifts[u]}]{\mat{A}} \ge \unishift =
  \rdeg[{\shifts[u]}]{\mat{B}}$ componentwise.
  
  Besides, $\mat{C} \sys = 0 \bmod (X^{\tau_1},\ldots,X^{\tau_\cdim})$ and
  since $\mat{C}$ has its $\shifts[u]$-pivot entries on the diagonal, it is
  $\shifts[u]$-reduced: by minimality of $\mat{A}$, we obtain
  $\rdeg[{\shifts[u]}]{\mat{A}} = \rdeg[{\shifts[u]}]{\mat{C}}$. Then, it is
  easily verified that $\mat{C}$ is in $\shifts[u]$-Popov form, hence $\mat{C} =
  \mat{A}$.
\end{proof}

\vspace{-0.2cm}
In particular, computing the $\shifts[s]$-Popov nullspace basis $\mat{B}$, when
its $\shifts$-pivot index, its $\shifts$-pivot degree, and
$\maxDeg\ge\deg(\mat{B})$ are known, can be done in $\softO{\rdim^{\expmatmul-1}
(\order + \cdim\maxDeg)}$ with $\order = \order_1 + \cdots + \order_\cdim$
using the order basis algorithm in~\cite{JeNeScVi16}.

As for Problem~\ref{pbm:modsys}, with Lemma~\ref{lem:mnb_solbas} this gives an
algorithm for computing $\msb$ \emph{and} the quotients $\mat{Q} = -\msb \sys /
\diag{\Modulus}$ when we know \emph{a priori} the $\shifts$-minimal degree
$\minDegs$ of $(\Modulus,\sys)$. Here, we would choose $\maxDeg =
\max(\minDegs) \ge \deg( [\msb | \mat{Q}])$: in some cases $\maxDeg =
\Theta(\degMod)$ and this has cost bound $\softO{\nbun^{\expmatmul-1} (\degMod
+ \nbeq \degMod)}$, which exceeds our target $\softO{\nbun^{\expmatmul-1}
\degMod}$. An issue is that $\mat{Q}$ has size $\bigO{\nbun \nbeq \degMod}$
when $\msb$ has columns of large degree; yet here we are not interested in
$\mat{Q}$. This can be solved using partial linearization to expand the columns
of large degree in $\msb$ into more columns of smaller degree as in the next
result, which holds in general for interpolation
bases~\cite[Lemma~4.2]{JeNeScVi16}.

\vspace{-0.2cm}
\begin{lem}
  \label{lem:modsys_parlin}
  Let $\Modulus \in \modSpace^\nbeq$ with entries having degrees $(\degMod_1,
  \ldots, \degMod_\nbeq)$. Let $\sys \in \sysSpace$ and $\shifts \in
  \shiftSpace$.  Furthermore, let $\minDegs = (\minDeg_1,\ldots,\minDeg_\nbun)$
  denote the $\shifts$-minimal degree of $(\Modulus,\sys)$.
  
  Writing $\degMod = \degMod_1 + \cdots + \degMod_\nbeq$, let $\degExp = \lceil
  \degMod / \nbun \rceil \ge 1$, and for $i\in\{1,\ldots,\nbun\}$ write
  $\minDeg_i = (\quoExp_i -1) \degExp + \remExp_i$ with $\quoExp_i \ge 1$ and
  $0 \le \remExp_i < \degExp$, and let $\expand{\nbun} = \quoExp_1 + \cdots +
  \quoExp_\nbun$. Define $\expand{\minDegs}\in \NN^{\expand{\nbun}}$ as
  \vspace{-0.2cm}
  \begin{equation}
    \label{eqn:expandMinDegs}
    \expand{\minDegs} = ( \underbrace{\degExp, \ldots, \degExp, \remExp_1}_{\quoExp_1}, \ldots,
    \underbrace{\degExp, \ldots, \degExp, \remExp_\nbun}_{\quoExp_\nbun} )
  \vspace{-0.2cm}
  \end{equation}
  and the expansion-compression matrix $\expandMat \in
  \polMatSpace[\expand{\nbun}][\nbun]$ as
  \vspace{-0.15cm}
  \begin{equation}
    \label{eqn:expandMat}
    \expandMat = 
    \left[\begin{smallmatrix}
      1 \\
      X^\degExp \\[-0.1cm]
      \vdots \\
      X^{(\quoExp_1-1)\degExp} \\[-0.2cm]
      & \;\;\; \ddots \\
      & & 1 \\
      & & X^\degExp \\
      & & \vdots \\
      & & X^{(\quoExp_\nbun-1)\degExp}
  \end{smallmatrix}\right] . 
  \vspace{-0.15cm}
  \end{equation}
  Let $\shifts[d] = - \expand{\minDegs} \in \shiftSpace[\expand{\nbun}]$ and
  $\mat{P} \in \polMatSpace[\expand{\nbun}]$ be the $\shifts[d]$-Popov solution
  basis for $(\Modulus,\expandMat \sys \bmod \Modulus)$.
  Then, $\mat{P}$ has $\shifts[d]$-pivot degree $\expand{\minDegs}$ and the
  $\shifts$-Popov solution basis for $(\Modulus,\sys)$ is the submatrix of
  $\mat{P} \expandMat$ formed by its rows at indices
  $\{\quoExp_1+\cdots+\quoExp_i, 1\le i\le \rdim\}$.
\end{lem}

\vspace{-0.1cm} This leads to Algorithm~\ref{algo:knownmindeg_modsys}, which
solves Problem~\ref{pbm:modsys} efficiently when the $\shifts$-minimal degree
$\minDegs$ is known \emph{a priori}.

\vspace{-0.15cm}
\begin{figure}[h!]
\centering
\fbox{\begin{minipage}{8.2cm}
\begin{algo} [\algoname{KnownDegPolModSys}]  \label{algo:knownmindeg_modsys}
  \normalfont

  ~\\
	\emph{Input:} 
    polynomials $\Modulus = (\modulus_1,\ldots,\modulus_\nbeq) \in \modSpace^\nbeq$,
    a matrix $\sys \in \sysSpace$ with $\deg(\matcol{\sys}{j}) < \deg(\modulus_j)$,
    a shift $\shifts\in\shiftSpace$,
    $\minDegs = (\minDeg_1,\ldots,\minDeg_\rdim)$ the $\shifts$-minimal degree of $(\Modulus,\sys)$.

  \smallskip
  \emph{Output:} the $\shifts$-Popov solution basis for $(\Modulus,\sys)$.
  \vspace{-0.1cm}
	\begin{enumerate}[{\bf 1.}] 
      \setlength{\itemsep}{0cm}
    \item $\degExp \leftarrow \lceil (\deg(\modulus_1) + \cdots + \deg(\modulus_\nbeq)) / \rdim \rceil$, \\
      $\quoExp_i \leftarrow \lfloor \minDeg_i / \degExp \rfloor + 1$ for $1 \le
      i \le \rdim$, 
      $\expand{\rdim} \leftarrow \quoExp_1 + \cdots +
      \quoExp_\rdim$, \\
      $\expand{\minDegs}$ as in~\eqref{eqn:expandMinDegs},
      $\expandMat$ as in~\eqref{eqn:expandMat},
      $\expand{\sys} \leftarrow \expandMat \sys \bmod \Modulus$
    \item $\shifts[u] \leftarrow (-\expand{\minDegs},-\maxDeg-1,\ldots,-\maxDeg-1)
      \in \shiftSpace[\expand{\nbun}+\nbeq]$ \\
      $\boldsymbol{\tau} \leftarrow (\deg(\modulus_j) + \maxDeg+1)_{1 \le j \le \nbeq}$
    \item $\expand{\msb} \leftarrow$ the $\shifts[u]$-Popov order basis for $\trsp{[\trsp{\expand{\sys}}|\diag{\Modulus}]}$
      and $\boldsymbol{\tau}$ \\
      $\mat{P}$ $\leftarrow$ the principal $\expand{\nbun} \times \expand{\nbun}$ submatrix of $\expand{\msb}$
    \item Return the submatrix of $\mat{P} \expandMat$ formed by the rows
      at indices $\quoExp_1+\cdots+\quoExp_i$ for $1\le i\le \rdim$
	\end{enumerate}
\end{algo}
\end{minipage}} 
\vspace{-0.3cm}
\end{figure}

\vspace{-0.3cm}
\begin{prop}
  \label{prop:algo:knownmindeg_modsys}
  Algorithm \algoname{KnownDegPolModSys} is correct. Writing $\degMod =
  \deg(\modulus_1) + \cdots + \deg(\modulus_\nbeq)$ and assuming $\degMod \ge
  \nbun \ge \nbeq$, it uses $\softO{ \nbun^{\expmatmul-1} \degMod }$ operations
  in $\field$.
\end{prop}
\vspace{-0.3cm}
\begin{proof}
  By Lemmas~\ref{lem:modsys_parlin} and~\ref{lem:mnb_solbas}, since
  $\min(-\expand{\minDegs}) > -\minDeg-1$ and $\shifts[u] =
  (-\expand{\minDegs},-\maxDeg-1,\ldots,-\maxDeg-1)$, the $-\minDegs$-Popov
  solution basis for $(\Modulus,\expand{\sys})$ is the principal
  $\expand{\nbun} \times \expand{\nbun}$ submatrix of the $\shifts[u]$-Popov
  nullspace basis $\mat{B}$ for
  $\trsp{[\trsp{\expand{\sys}}|\diag{\Modulus}]}$, and $\mat{B}$ has
  $\shifts[u]$-pivot index $\{1,\ldots,\expand{\nbun}\}$, $\shifts[u]$-pivot
  degree $\expand{\minDegs}$, and $\deg(\mat{B}) \le \minDeg$. Then, by
  Lemma~\ref{lem:mnb_known_mindeg}, $\mat{B}$ is formed by the first
  $\expand{\nbun}$ rows of $\expand{\mat{P}}$ at Step~\textbf{3}, hence
  $\mat{P}$ is the $\shifts[d]$-Popov solution basis for $(\Modulus,\sys)$. The
  correctness then follows from Lemma~\ref{lem:modsys_parlin}. 

  Since $\sumVec{\minDegs} \le \degMod$, $\expandMat$ has $\expand{\rdim} \le
  2\rdim$ rows and $\expandMat \sys \bmod \Modulus$ can be computed in
  $\softO{\nbun \degMod}$ operations using fast polynomial
  division~\cite{vzGathen13}. The cost bound of Step~\textbf{3} follows from
  \cite[Theorem~1.4]{JeNeScVi16} since $\tau_1 + \cdots + \tau_\nbeq = \degMod
  + \nbeq(1+\lceil \degMod / \nbun \rceil) \in \bigO{\degMod}$.
\end{proof}

\subsection{The case of one equation}
\label{subsec:modsys_one}

We now present our main new ingredients, focusing on the case $\nbeq=1$. First,
we show that when the shift $\shifts$ has a small \emph{amplitude} $\amp =
\max(\shifts)-\min(\shifts)$, one can solve Problem~\ref{pbm:modsys} via an
order basis computation at small order.

\vspace{-0.15cm}
\begin{lem}
  \label{lem:modsys_one_small}
  Let $\modulus \in \modSpace$, $\shifts \in \shiftSpace$, and $\sys \in
  \polMatSpace[\nbun][1]$ with $\deg(\sys) < \deg(\modulus) = \degMod$. Then,
  for any $\tau \ge \amp+2\degMod$, the $\shifts$-Popov solution basis for
  $(\modulus,\sys)$ is the principal $\nbun \times \nbun$ submatrix of the
  $\shifts[u]$-Popov order basis for $\trsp{ [\trsp{\sys} | \modulus] }$ and
  $\tau$, with $\shifts[u] = (\shifts,\min(\shifts)) \in \shiftSpace[\nbun+1]$.
\end{lem}
\begin{proof}
  Let $\mat{A}
    = \left[ \begin{smallmatrix}
        \msb & \col{q} \\
        \row{p} & q 
    \end{smallmatrix} \right]$
  denote the $\shifts[u]$-Popov order basis for $\trsp{
    [\trsp{\sys} | \modulus] }$ and $\tau$,
  where $\mat{P} \in \polMatSpace[\nbun]$
  and $q \in \polRing$. Consider $\mat{B} = [\bar{\mat{P}} | \bar{\col{q}} ]$
  the $\shifts[u]$-Popov nullspace basis of $\trsp{[\trsp{\sys} | \modulus]}$:
  thanks to Lemma~\ref{lem:mnb_solbas}, it is enough to prove that $\mat{B} =
  [\mat{P}|\col{q}]$.

  First, we have $\rdeg{\row{p}} \le \deg(q)$ by choice of $\shifts[u]$, so
  that $q \modulus \neq 0$ implies $\deg(\row{p} \sys + q \modulus) = \deg(q) +
  \degMod$. Since $\row{p} \sys + q \modulus = 0 \bmod X^{\tau}$, this gives
  $\deg(q) + \degMod \ge \tau$. This also shows that the $\shifts[u]$-pivot
  entries of $\mat{B}$ are located in $\bar{\mat{P}}$.

  Then, since the sum of the $\shifts[u]$-pivot degrees of $\mat{A}$ is at most
  $\tau$, the sum of the $\shifts$-pivot degrees of $\mat{P}$ is at most
  $\degMod$; with $[\mat{P}|\col{q}]$ in $\shifts[u]$-Popov form, this gives
  $\deg(\col{q}) < \degMod + \amp \le \tau-\degMod$. We obtain $\deg(\mat{P}
  \sys + \col{q} \modulus) < \tau$, so that $\mat{P} \sys + \col{q} \modulus =
  0$.
  Thus, the minimality of $\mat{B}$ and $\mat{A}$ gives the conclusion.
\end{proof}

\vspace{-0.15cm}
When $\amp \in \bigO{\degMod}$, this gives a fast solution to our problem. In
what follows, we present a divide-and-conquer approach on $\amp$, with base
case $\amp \in \bigO{\degMod}$.

We first give an overview, assuming $\shifts$ is non-decreasing. A key
ingredient is that when $\amp$ is large compared to $\degMod$, then $\msb$ has
a lower block triangular shape, since it is in $\shifts$-Popov form with sum of
$\shifts$-pivot degrees $\sumVec{\minDegs} \le \degMod$. Typically, if
$\shift{i+1} - \shift{i} \ge \degMod$ for some $i$ then $\msb =
\left[\begin{smallmatrix} \msb^{(1)} & \mat{0} \\ \anyMat & \msb^{(2)}
\end{smallmatrix}\right]$ with $\msb^{(1)} \in \polMatSpace[i]$.
Even though the block sizes are unknown in general, we show that they can be
revealed efficiently along with $\minDegs$ by a divide-and-conquer algorithm,
as follows.

First, we use a recursive call with the first $j$ entries $\shifts^{(0)}$ of
$\shifts$ and $\sys^{(0)}$ of $\sys$, where $j$ is such that
$\amp[\shifts^{(0)}]$ is about half of $\amp$. This reveals the first $i\le j$
entries $\minDegs^{(1)}$ of $\minDegs$ and the first $i$ rows $[\msb^{(1)} |
\mat{0}]$ of $\msb$, with $\msb^{(1)} \in \polMatSpace[i]$. A central point is
that $\amp[\shifts^{(2)}]$ is about half of $\amp$ as well, where
$\shifts^{(2)}$ is the tail of $\shifts$ starting at the entry $i+1$.

Then, knowing the degrees $\minDegs^{(1)}$ allows us to set up an order basis
computation that yields a \emph{residual}, that is, a column $\mat{G} \in
\polMatSpace[(\nbun-i)][1]$ and a modulus $\modulus[n]$ such that we can
continue the computation of $\msb$ using a second recursive call, which
consists in computing the $\shifts^{(2)}$-Popov solution basis for
$(\modulus[n],\mat{G})$. From these two calls we obtain $\minDegs$, and then
we recover $\msb$ using Algorithm~\ref{algo:knownmindeg_modsys}.

Now we present the details. We fix $\sys \in \polMatSpace[\nbun][1]$, $\modulus
\in \modSpace$ with $\degMod = \deg(\modulus) > \deg(\sys)$, $\shifts \in
\shiftSpace$, $\msb$ the $\shifts$-Popov solution basis for $(\modulus,\sys)$,
and $\minDegs$ its $\shifts$-pivot degree. In what follows, $\perm =
(\pi_1,\ldots,\pi_\nbun)$ is any permutation of $\{1,\ldots,\nbun\}$ such that
$(\shift{\pi_1},\ldots,\shift{\pi_\nbun})$ is non-decreasing.

Then, for $\shifts[t] = (t_1,\ldots,t_\nbun) \in \ZZ^\nbun$ we write
$\subVec{\shifts[t]}{i}{j}$ for the subtuple of $\shifts[t]$ formed by its
entries at indices $\{\pi_i,\ldots,\pi_j\}$, and for a matrix $\mat{M} \in
\polMatSpace[\nbun]$ we write $\subMat{\mat{M}}{i}{j}{k}{l}$ for the submatrix
of $\mat{M}$ formed by its rows at indices $\{\pi_i,\pi_{i+1},\ldots,\pi_j\}$
and columns at indices $\{\pi_k,\pi_{k+1},\ldots,\pi_l\}$. The main ideas in
this subsection can be understood by focusing on the case of a non-decreasing
$\shifts$, taking $\pi_i = i$ for all $i$: then we have
$\subVec{\shifts[t]}{i}{j} = (t_i,t_{i+1},\ldots,t_j)$ and
$\subMat{\mat{M}}{i}{j}{k}{l} = (\mat{M}_{u,v})_{i\le u\le j,k\le v\le l}$.

We now introduce the notion of splitting index, which will help us to locate
zero blocks in $\msb$.

\vspace{-0.15cm}
\begin{dfn}[Splitting index]
  Let $\shifts[d] \in \NN^\nbun$, $\shifts[t] \in \ZZ^\nbun$ and
  $\perm[{\shifts[t]}] = (\mu_i)_i$. Then, $i \in \{1,\ldots,\nbun-1\}$ is a
  \emph{splitting index} for $(\shifts[d],\shifts[t])$ if $d_{\mu_j} +
  \shift[t]{\mu_j} - \shift[t]{\mu_{i+1}} < 0$ for all $j \in \{1,\ldots,i\}$.
\end{dfn}

\vspace{-0.15cm}
In particular, if $i$ is a splitting index for $(\minDegs,\shifts)$, then we
have $[\subMat{\msb}{}{i}{}{i} | \subMat{\msb}{}{i}{i+1}{}] =
[\subMat{\msb}{}{i}{}{i} | \mat{0}]$. Our algorithm first looks for such a
splitting index, and then uses $\subMat{\msb}{}{i}{i+1}{} = \mat{0}$ to split
the problem into two subproblems of dimensions $i$ and $\nbun-i$.

To find a splitting index, we rely on the following property: if
$(\shifts[d],\shifts[t])$ does not admit a splitting index, then
$\sumVec{\shifts[d]} > \amp[\shifts[t]]$. This allows us to partition $\shifts$
into $\ell$ subtuples which all contain a splitting index, as follows.

Given $\alpha \in \ZZp$ we let $\ell = 1 + \lfloor \amp / \alpha \rfloor$ and
we consider the subtuples $\shifts_1,\ldots,\shifts_\ell$ of~$\shifts$ where
$\shifts_k$ consists of the entries of $\shifts$ in $\{\min(\shifts) + (k-1)
\alpha, \ldots, \min(\shifts) + k\alpha-1\}$; this gives a subroutine
$\algoname{Partition}(\shifts,\alpha) = (\shifts_1,\ldots,\shifts_\ell)$. Now
we take $\alpha \ge 2\degMod$ and we assume $\shift{\pi_{i+1}} - \shift{\pi_i}
\le \degMod$ for $1 \le i < \nbun$ without loss of generality \cite[Appendix
A]{JeNeScVi16}. Then, for $1\le k<\ell$, since $\sumVec{\minDegs} \le \degMod$
and $\amp[\shifts[t]] \ge \degMod$ with $\shifts[t] =
(\shifts_k,\min(\shifts_{k+1}))$, by the above remark $\shifts_k$ contains a
splitting index for $(\minDegs,\shifts)$.

Still, we do not know in advance which entries of $\shifts_k$ correspond to
splitting indices for $(\minDegs,\shifts)$. Thus we recursively compute the
$\shifts$-Popov solution basis $\msb^{(0)}$ for
$\shifts_1,\ldots,\shifts_{\ell/2}$, and we are now going to prove that this
gives us a splitting index which divides the computation into two subproblems,
the first of which has been already solved by computing $\msb^{(0)}$.

\vspace{-0.2cm}
\begin{lem}
  \label{lem:one:splitting_index}
  Let $j \in \{2,\ldots,\nbun\}$, $\shifts^{(0)} = \subVec{\shifts}{}{j}$,
  $\msb^{(0)}$ be the $\shifts^{(0)}$-Popov solution basis for
  $(\modulus,\subVec{\sys}{}{j})$, and $\minDegs^{(0)}$ be its
  $\shifts^{(0)}$-pivot degree. Suppose that there is a splitting index $i \le
  j$ for $(\minDegs^{(0)},\shifts^{(0)})$. Let $\msb^{(1)} \in \polMatSpace[i]$
  be the $\shifts^{(1)}$-Popov solution basis for
  $(\modulus,\subVec{\sys}{}{i})$ with $\shifts^{(1)} = \subVec{\shifts}{}{i}$,
  and let $\minDegs^{(1)}$ be its $\shifts^{(1)}$-pivot degree. Then $i$ is a
  splitting index for $(\minDegs,\shifts)$ and $\subMat{\msb}{}{i}{}{i} =
  \msb^{(1)} = \subMat{\msb^{(0)}}{}{i}{}{i}$, hence $\subVec{\minDegs}{}{i} =
  \minDegs^{(1)} = \subVec{\minDegs^{(0)}}{}{i}$ (where $\msb^{(0)}$ and
  $\minDegs^{(0)}$ are indexed by $\{\pi_1,\ldots,\pi_j\}$ sorted
  increasingly).
\end{lem}
\begin{proof}
\vspace{-0.2cm}
  Since $i$ is a splitting index for $(\minDegs^{(0)},\shifts^{(0)})$ we have
  $[\subMat{\msb^{(0)}}{}{i}{}{i} | \subMat{\msb^{(0)}}{}{i}{i+1}{}] = [\mat{Q}
  | \mat{0} ]$ for some $\mat{Q} \in \polMatSpace[i]$. Now, for any $\mat{B}
  \in \polMatSpace[\nbun][\nbun]$ with $[\subMat{\mat{B}}{}{i}{}{i} |
  \subMat{\mat{B}}{}{i}{i+1}{}] = [\msb^{(1)} | \mat{0}]$,
  $\subMat{\mat{B}}{}{i}{}{}$ is in $\shifts$-Popov form with its rows being
  solutions for $(\Modulus,\sys)$. Then, by minimality of $\msb$,
  $\subMat{\msb}{}{i}{}{}$ has $\shifts$-pivot degree at most $\minDegs^{(1)}$
  componentwise, so that $i$ is also a splitting index for
  $(\minDegs,\shifts)$, and in particular $[\subMat{\msb}{}{i}{}{i} |
  \subMat{\msb}{}{i}{i+1}{}] = [\mat{R} | \mat{0} ]$ for some $\mat{R} \in
  \polMatSpace[i]$. It remains to prove that $\mat{Q} = \mat{R} = \msb^{(1)}$.

  Since $\mat{R} \subVec{\sys}{}{i} = 0 \bmod \modulus$ and $\mat{R} =
  \subMat{\msb}{}{i}{}{i}$ is in $\shifts^{(1)}$-Popov form, proving that all
  solutions $\row{p} \in \polMatSpace[1][i]$ for
  $(\modulus,\subVec{\sys}{}{i})$ are in the row space of $\mat{R}$ is enough
  to obtain $\mat{R} = \msb^{(1)}$. Since $\row{q} \in \polMatSpace[1][\nbun]$
  defined by $[\subVec{\row{q}}{}{i} | \subVec{\row{q}}{i+1}{}] = [\row{p} |
  \mat{0}]$ is a solution for $(\modulus,\sys)$, $\row{q} =
  \boldsymbol{\lambda} \msb$ for some $\boldsymbol{\lambda} \in
  \polMatSpace[1][\nbun]$. Now $\msb$ is nonsingular, thus
  $\subMat{\msb}{}{i}{i+1}{} = \mat{0}$ implies that
  $[\subVec{\boldsymbol{\lambda}}{}{i} | \subVec{\boldsymbol{\lambda}}{i+1}{}]
  = [\boldsymbol{\mu} | \row{0}]$ with $\boldsymbol{\mu} \in
  \polMatSpace[1][i]$, hence $\row{p} = \subVec{\row{q}}{}{i} =
  \subVec{\boldsymbol{\lambda}}{}{i} \subMat{\msb}{}{i}{}{i} +
  \subVec{\boldsymbol{\lambda}}{i+1}{} \subMat{\msb}{i+1}{}{}{i} =
  \boldsymbol{\mu} \mat{Q}$. Similar arguments give $\mat{Q} = \msb^{(1)}$.
\end{proof}

\vspace{-0.2cm}
The next two lemmas show that knowing $\minDegs^{(1)}$, which is
$\subVec{\minDegs}{}{i}$, allows us to compute a so-called \emph{residual}
$(\modulus[n],\mat{G})$ from which we can complete the computation of
$\minDegs$ and $\msb$.

\vspace{-0.15cm}
\begin{lem}
  \label{lem:one:order_basis_shape}
  Let $\shifts^{(2)} = \subVec{\shifts}{i+1}{}$, $\shifts[d] =
  -\minDegs^{(1)}+\min(\shifts^{(2)})-2\degMod \in \shiftSpace[i]$, $\shifts[v]
  \in \shiftSpace$ be such that $[\subVec{\shifts[v]}{}{i} |
  \subVec{\shifts[v]}{i+1}{}] = [\shifts[d] | \shifts^{(2)}]$, and $\shifts[u]
  = (\shifts[v],\min(\shifts[d])) \in \shiftSpace[\nbun+1]$.  Let 
  $\left[ \begin{smallmatrix}
    \mat{A} & \col{q} \\
    \row{p} & q 
  \end{smallmatrix} \right]$
  be the $\shifts[u]$-Popov order basis for $\trsp{[\trsp{\sys} | \modulus]}$
  and $2\degMod$, where $\mat{A} \in \polMatSpace[\nbun]$ and $q \in \polRing$.
  Then we have $\deg(q) \ge \degMod$, $\subMat{\mat{A}}{}{i}{i+1}{} = \mat{0}$,
  $\subVec{\row{p}}{i+1}{} = \mat{0}$, and $[ \subMat{\mat{A}}{}{i}{}{i} |
  \subVec{\row{q}}{}{i} ] = [ \msb^{(1)} | \col{q}^{(1)} ]$ with $\col{q}^{(1)}
  = -\msb^{(1)} \subVec{\sys}{}{i} / \modulus$.
\end{lem}
\begin{proof}
  \vspace{-0.1cm}
  Since $\shifts[u] = (\shifts[v],\min(\shifts[v]))$ we have $\deg(\row{p}) \le
  \deg(q)$, and since $\deg(\sys) < \deg(\modulus)$ the degree of $\row{p} \sys
  + q \modulus$ is $\deg(q) + \degMod$; then $\row{p} \sys + q \modulus = 0
  \bmod X^{2\degMod}$ implies $\deg(q) + \degMod \ge 2\degMod$. Now, since
  $\mat{A}$ is in $\shifts[v]$-Popov form and $\deg(\mat{A}) \le 2\degMod -
  \deg(q) < 2\degMod$, from $\min(\shifts^{(2)}) \ge \max(\shifts[d]) +
  2\degMod$ we get $\subMat{\mat{A}}{}{i}{i+1}{} = \mat{0}$. Besides,
  $\subVec{\row{p}}{i+1}{} = \mat{0}$ since either $\deg(q)<2\degMod$ and then
  $\min(\shifts^{(2)}) > \min(\shifts[d]) + \deg(q)$, or $\mat{A}$ is the
  identity matrix and then $\row{p} = \row{0}$.

  Furthermore, by Lemma~\ref{lem:mnb_solbas} $[\msb^{(1)} | \col{q}^{(1)}]$ is the
  $(\shifts[d],\min(\shifts[d]))$-Popov nullspace basis for
  $\trsp{[\trsp{\subVec{\sys}{}{i}} | \modulus]}$, with
  $(\shifts[d],\min(\shifts[d]))$-pivot index $\{1,\ldots,i\}$,
  $(\shifts[d],\min(\shifts[d]))$-pivot degree $\minDegs^{(1)}$ and degree at
  most $\max(\minDegs^{(1)})$. Then, as in the proof of
  Lemma~\ref{lem:mnb_known_mindeg}, one can show that $[
  \subMat{\mat{A}}{}{i}{}{i} | \subVec{\row{q}}{}{i} ] = [ \msb^{(1)} |
  \col{q}^{(1)} ]$.
\end{proof}

\vspace{-0.15cm}
Thus, up to row and column permutations this order basis is
$\left[\begin{smallmatrix}
    \msb^{(1)} & \mat{0} & \col{q}^{(1)} \\
    \anyMat & \msb^{(2)} & \anyMat \\ \anyMat & \mat{0} & q
\end{smallmatrix}\right]$
with $\msb^{(2)} = \subMat{\mat{A}}{i+1}{}{i+1}{} \in \polMatSpace[(\nbun-i)]$ in
$\shifts^{(2)}$-Popov form; let $\minDegs^{(2)}$ denote its
$\shifts^{(2)}$-pivot degree.

\vspace{-0.1cm}
\begin{lem}
  \label{lem:one:obtain_degrees}
  Let $\modulus[n] = X^{-2\degMod} (\subVec{\row{p}}{i+1}{}
  \subVec{\sys}{i+1}{} + q \modulus) \in \polRing$ and $\mat{G} = X^{-2\degMod}
  (\subMat{\mat{A}}{i+1}{}{}{} \sys + \subVec{\col{q}}{i+1}{}
  \modulus) \in \polMatSpace[(\nbun-i)][1]$. Then, $\deg(\mat{G}) <
  \deg(\mathfrak{n}) \le \degMod - \sumVec{\minDegs^{(1)}} -
  \sumVec{\minDegs^{(2)}}$.
  Let $\msb^{(3)}$ be the $\shifts[t]$-Popov solution basis for
  $(\modulus[n],\mat{G})$ with $\shifts[t] = \rdeg[\shifts^{(2)}]{\msb^{(2)}}$
  and $\minDegs^{(3)}$ be its $\shifts[t]$-pivot degree. Then,
  $(\subVec{\minDegs}{}{i},\subVec{\minDegs}{i+1}{}) =
  (\minDegs^{(1)},\minDegs^{(2)} + \minDegs^{(3)})$.
\end{lem}
\begin{proof}
  \vspace{-0.1cm}
  The sum $\sumVec{\minDegs^{(1)}} + \sumVec{\minDegs^{(2)}} + \deg(q)$ of the
  $\shifts[u]$-pivot degrees of 
  $\left[ \begin{smallmatrix}
    \mat{A} & \col{q} \\
    \row{p} & q 
  \end{smallmatrix} \right]$
  is at most the order $2\degMod$.
  Thus, we have $\deg(\modulus[n]) = \deg(q) - \degMod \le \degMod -
  \sumVec{\minDegs^{(1)}} - \sumVec{\minDegs^{(2)}}$,
  $\deg(\subMat{\mat{A}}{i+1}{}{}{i}) < \sumVec{\minDegs^{(1)}} \le \degMod$,
  $\deg(\subMat{\mat{A}}{i+1}{}{i+1}{}) \le \sumVec{\minDegs^{(2)}} \le
  \degMod$, and $\deg(\subVec{\col{q}}{i+1}{}) < \deg(q)$. This implies
  $\deg(\mat{G}) < \deg(q) - \degMod = \deg(\modulus[n])$.

  Let $\col{q}^{(3)} = -\msb^{(3)} \mat{G} / \modulus[n]$ and $t =
  \rdeg[{\shifts[u]}]{[\row{p} | q]} = \deg(q) + \min(\shifts[d])
  \le \min(\shifts^{(2)}) \le \min(\shifts[t])$. By Lemma~\ref{lem:mnb_solbas},
  $[\msb^{(3)} | \col{q}^{(3)}]$ is the $(\shifts[t],t)$-Popov nullspace basis
  for $\trsp{[\trsp{\mat{G}} | \modulus[n]]}$. Defining $\mat{B} \in
  \polMatSpace[\nbun]$ and $\col{c} \in \polMatSpace[\nbun][1]$ by
  $\left[\begin{smallmatrix} \subMat{\mat{B}}{}{i}{}{i} &
      \subMat{\mat{B}}{}{i}{i+1}{} & \subVec{\col{c}}{}{i} \\
      \subMat{\mat{B}}{i+1}{}{}{i} & \subMat{\mat{B}}{i+1}{}{i+1}{} &
      \subVec{\col{c}}{i+1}{}
  \end{smallmatrix} \right] =
  \left[ \begin{smallmatrix}
    \idMat[] & \mat{0} & \mat{0} \\
    \mat{0} & \msb^{(3)} & \col{q}^{(3)}
  \end{smallmatrix} \right]$,
  then
  $\left[\begin{smallmatrix} \mat{B} \;\; \col{c} \end{smallmatrix} \right]
  \left[ \begin{smallmatrix}
    \mat{A} & \col{q} \\
    \row{p} & q 
  \end{smallmatrix} \right]$
  is a $\shifts[u]$-minimal nullspace basis of $\trsp{[\trsp{\sys} |
  \modulus]}$~\cite[Theorem~3.9]{ZhLaSt12}. Thus Lemma~\ref{lem:mnb_solbas}
  implies that
  $\bar{\msb} = \left[\begin{smallmatrix} \mat{B} \;\; \col{c} \end{smallmatrix} \right]
  \left[ \begin{smallmatrix}
    \mat{A}  \\
    \row{p}
  \end{smallmatrix} \right]$
  is a $\shifts[v]$-minimal solution basis for $(\modulus,\sys)$.

  It is easily checked that $\msb$ is in $\shifts[v]$-Popov form, so that the
  $\shifts[v]$-Popov form of $\bar{\msb}$ is $\msb$ and its $\shifts[v]$-pivot
  degree is $\minDegs$. Besides $\left[\begin{smallmatrix}
      \subMat{\bar{\msb}}{}{i}{}{i} & \subMat{\bar{\msb}}{}{i}{i+1}{} \\
      \subMat{\bar{\msb}}{i+1}{}{}{i} & \subMat{\bar{\msb}}{i+1}{}{i+1}{}
  \end{smallmatrix}\right]
      = \left[\begin{smallmatrix} \msb^{(1)} & \mat{0} \\ \msb^{(3)} \mat{A}_{2,1} +
      \col{q}^{(3)} \mat{A}_{3,1} & \msb^{(3)} \msb^{(2)}
  \end{smallmatrix}\right]$,
  so that 
  $(\subVec{\minDegs}{}{i},\subVec{\minDegs}{i+1}{}) = (\minDegs^{(1)},\minDegs^{(2)}+\minDegs^{(3)})$
  \cite[Section 3]{JeNeScVi16}.
\end{proof}
\vspace{-0.2cm}

This results in Algorithm~\ref{algo:modsysone}. It takes as input $\alpha$
which dictates the amplitude of the subtuples that partition $\shifts$; as
mentioned above, the initial call can be made with $\alpha = 2\degMod$.

\begin{figure}[h!]
\centering
\fbox{\begin{minipage}{8.2cm}
\begin{algo} [\algoname{PolModSysOne}] 
  \label{algo:modsysone}
  \normalfont
  ~\\
	\emph{Input:} 
      a polynomial $\modulus \in \modSpace$ of degree $\degMod$,
      a column $\sys \in \polMatSpace[\nbun][1]$ with $\deg(\sys) < \deg(\modulus)$,
      a shift $\shifts\in\shiftSpace$,
      a parameter $\alpha \in \ZZp$ with $\alpha \ge 2\degMod$.

  \smallskip
  \emph{Output:} the $\shifts$-Popov solution basis
  for $(\modulus,\sys)$ and the $\shifts$-minimal degree $\minDegs$ of
  $(\modulus,\sys)$.
  \vspace{-0.1cm}
	\begin{enumerate}[{\bf 1.}] 
      \setlength{\itemsep}{0cm}
    \item If $\amp \le 2\alpha$:
      \vspace{-0.1cm}
      \begin{enumerate}[{\bf a.}]
      \setlength{\itemsep}{0cm}
        \item $\mat{A}$ $\leftarrow$ the $(\shifts,\min(\shifts))$-Popov order
          basis for $\trsp{[\trsp{\sys} | \modulus]}$ and $2\alpha + 2\degMod$; return
          the principal $\nbun \times \nbun$ submatrix of $\mat{A}$ and the
          degrees of its diagonal entries
      \end{enumerate}
      \vspace{-0.1cm}
    \item Else: \hfill \texttt{/* $\ell = 1 + \lfloor \amp / \alpha \rfloor \ge 3$ */}
      \vspace{-0.1cm}
      \begin{enumerate}[{\bf a.}]
      \setlength{\itemsep}{0cm}
        \item $(\shifts_1,\ldots,\shifts_\ell) \leftarrow \algoname{Partition}(\shifts,\alpha)$, \\
          $j \leftarrow$ sum of the lengths of $\shifts_1,\ldots,\shifts_{\lceil \ell/2 \rceil}$, $\shifts^{(0)} \leftarrow \subVec{\shifts}{}{j}$, \\
          $(\msb^{(0)},\minDegs^{(0)}) \leftarrow \algoname{PolModSysOne}( \modulus, \subVec{\sys}{}{j}, \shifts^{(0)}, \alpha )$
        \item $i \leftarrow$ the largest splitting index for $(\minDegs^{(0)},\shifts^{(0)})$,
          $\minDegs^{(1)} \leftarrow \subVec{\minDegs^{(0)}}{}{i}$, $\shifts^{(2)} \leftarrow \subVec{\shifts}{i+1}{}$,
          $\shifts[d] = -\minDegs^{(1)}+\min(\shifts^{(2)})-2\degMod$,
          $\shifts[v] \in \shiftSpace$ with $[\subVec{\shifts[v]}{}{i} |
          \subVec{\shifts[v]}{i+1}{}] \leftarrow [\shifts[d] | \shifts^{(2)}]$,
          $\shifts[u] = (\shifts[v],\min(\shifts[d]))$
        \item $\left[ \begin{smallmatrix}
            \mat{A} & \col{q} \\
            \row{p} & q 
          \end{smallmatrix} \right]
          \leftarrow$ $\shifts[u]$-Popov order basis for $\trsp{[\trsp{\sys} | \modulus]}$
          and $2\degMod$,
          $\minDegs^{(2)} \leftarrow$ the $\shifts^{(2)}$-pivot degree of $\subMat{\mat{A}}{i+1}{}{i+1}{}$ \\
          $\mat{G} \leftarrow X^{-2\degMod}
          (\subMat{\mat{A}}{i+1}{}{}{} \sys + \subVec{\col{q}}{i+1}{}
          \modulus)$, \\
          $\modulus[n] \leftarrow X^{-2\degMod} (\subVec{\row{p}}{i+1}{}
          \subVec{\sys}{i+1}{} + q \modulus)$. 
        \item $\shifts[t] \leftarrow \shifts^{(2)} + \minDegs^{(2)} = \rdeg[\shifts^{(2)}]{\subMat{\mat{A}}{i+1}{}{i+1}{}}$, \\
          $(\msb^{(3)},\minDegs^{(3)}) \leftarrow \algoname{PolModSysOne}( \modulus[n], \mat{G}, \shifts[t], \alpha )$
        \item $\minDegs \in \NN^\nbun$ with $(\subVec{\minDegs}{}{i},\subVec{\minDegs}{i+1}{}) \leftarrow
          (\minDegs^{(1)},\minDegs^{(2)} + \minDegs^{(3)})$  \\
          $\msb \leftarrow \algoname{KnownDegPolModSys}(\modulus,\sys,\shifts,\minDegs)$
        \item Return $(\msb,\minDegs)$
      \end{enumerate}
      \vspace{-0.1cm}
	\end{enumerate}
\end{algo}
\end{minipage}} 
  \vspace{-0.5cm}
\end{figure}

\vspace{-0.2cm}
\begin{prop}
  \label{prop:modsys_one}
  Algorithm \algoname{PolModSysOne} is correct and uses
  $\softO{\nbun^{\expmatmul-1} \degMod}$ operations in $\field$.
\end{prop}
\begin{proof}
  \vspace{-0.1cm}
  The correctness follows from the results in this subsection.
  By~\cite[Theorem~1.4]{JeNeScVi16}, each leaf of the recursion
  at Step~\textbf{1.a} in dimension $\nbun$ uses $\softO{\nbun^{\expmatmul-1}
  \alpha}$ operations.

  Running the algorithm with initial input $\alpha = 2\degMod$, the recursive
  tree has depth $\bigO{\log(\ell)} = \bigO{\log(1 + \amp / 2\degMod)}$, with
  $\amp / 2\degMod \in \bigO{\nbun^2}$ \cite[Appendix A]{JeNeScVi16}. All
  recursive calls are for a modulus of degree $\degMod < \alpha$.  The order
  basis computation at Step~\textbf{2.c} uses $\softO{\nbun^{\expmatmul-1}
  \degMod}$ operations; the computation of $\mat{G}$ and $\modulus[n]$ at
  Step~\textbf{2.c} can be done in time $\softO{\nbun^{\expmatmul-1} \degMod}$
  using partial linearization as in Lemma~\ref{lem:residual} below;
  Step~\textbf{2.e} uses $\softO{\nbun^{\expmatmul-1} \degMod}$ operations by
  Proposition~\ref{prop:algo:knownmindeg_modsys}.

 On a given level of the tree, the sum of the dimensions of the column vector
 in input of each sub-problem is in $\bigO{\nbun}$. Since $a^{\expmatmul-1} +
 b^{\expmatmul-1} \le (a+b)^{\expmatmul-1}$ for all $a,b>0$, each level of the
 tree uses a total of $\softO{\nbun^{\expmatmul-1} \alpha}$ operations.
 \vspace{-0.1cm}
\end{proof}

\subsection{Fast divide-and-conquer algorithm}
\label{subsec:algo:modsys}

Now that we have an efficient algorithm for $\nbeq=1$, our main algorithm uses
a divide-and-conquer approach on $\nbeq$. Similarly to~\cite[Algorithm
1]{JeNeScVi16}, from the two bases obtained recursively we first deduce the
$\shifts$-minimal degree $\minDegs$, and then we use this knowledge to compute
$\msb$ with Algorithm~\ref{algo:knownmindeg_modsys}. When $\degMod =
\deg(\modulus_1) + \cdots + \deg(\modulus_\nbeq) \in \bigO{\nbun}$, we rely on
the algorithm \algoname{LinearizationMIB} in~\cite[Algorithm~9]{JeNeScVi15}.

\begin{figure}[h!]
\centering
\fbox{\begin{minipage}{8.2cm}
\begin{algo} [\algoname{PolModSys}] 
\label{algo:modsys}
  \normalfont
~\\
\emph{Input:} 
    polynomials $\Modulus = (\modulus_1,\ldots,\modulus_\nbeq) \in
    \modSpace^\nbeq$, a matrix $\sys \in \sysSpace$ with
    $\deg(\matcol{\sys}{j}) < \deg(\modulus_j)$, a shift
    $\shifts\in\shiftSpace$.

\smallskip
\emph{Output:} the $\shifts$-Popov solution basis
for $(\Modulus,\sys)$ and the $\shifts$-minimal degree $\minDegs$ of $(\Modulus,\sys)$.
\vspace{-0.1cm}
\begin{enumerate}[{\bf 1.}] 
  \setlength{\itemsep}{0cm}
  \item If $\degMod = \deg(\modulus_1) + \cdots + \deg(\modulus_\nbeq) \le \nbun$:
    \vspace{-0.1cm}
    \begin{enumerate}[{\bf a.}]
      \setlength{\itemsep}{0cm}
      \item Build $\evMat \in \matSpace[\rdim][\degMod]$ and $\mulmat \in \matSpace[\degMod]$ as in Section~\ref{subsec:modsys}
      \item Return $\algoname{LinearizationMIB}(\evMat,\mulmat,\shifts,2^{\lceil\log_2(\degMod)\rceil})$
    \end{enumerate}
    \vspace{-0.1cm}
  \item Else if $\nbeq = 1$:
    Return $\algoname{PolModSysOne}(\modulus_1,\sys,\shifts,2\degMod)$
  \item Else:
    \vspace{-0.1cm}
    \begin{enumerate}[{\bf a.}]
      \setlength{\itemsep}{0cm}
      \item $\Modulus^{(1)}, \sys^{(1)} \leftarrow (\modulus_1,\ldots,\modulus_{\lfloor \nbeq/2 \rfloor})$, $\matcol{\sys}{1\ldots \lfloor\nbeq/2\rfloor}$ \\
        $\Modulus^{(2)}, \sys^{(2)} \leftarrow (\modulus_{\lfloor \cdim/2 \rfloor+1},\ldots,\modulus_\cdim)$, $\matcol{\sys}{\lfloor\nbeq/2\rfloor+1 \ldots \nbeq}$
      \item $\msb^{(1)},\minDegs^{(1)}$ $\leftarrow$ $\algoname{PolModSys}(\Modulus^{(1)},\sys^{(1)},\shifts)$
      \item $\mat{R} \leftarrow$ $\msb^{(1)} \sys^{(2)} \bmod \Modulus^{(2)}$
      \item $\msb^{(2)},\minDegs^{(2)}$ $\leftarrow$ $\algoname{PolModSys}(\Modulus^{(2)},\mat{R},\rdeg[\shifts]{\msb^{(1)}})$
      \item $\msb \leftarrow \algoname{KnownDegPolModSys}(\Modulus,\sys,\shifts,\minDegs^{(1)}+\minDegs^{(2)})$ 
      \item Return $(\msb,\minDegs^{(1)}+\minDegs^{(2)})$
    \end{enumerate}
    \vspace{-0.1cm}
\end{enumerate}
\end{algo}
\end{minipage}} 
\vspace{-0.4cm}
\end{figure}

The computation of the so-called \emph{residual} at Step \textbf{3.c} can be
done efficiently using partial linearization, as follows.

\vspace{-0.15cm}
\begin{lem}
\label{lem:residual}
Let $\Modulus = (\modulus_j)_j \in \modSpace^\cdim$, $\mat{P} \in
\polMatSpace[\rdim]$, $\mat{F} \in \polMatSpace[\rdim][\cdim]$ with $\rdim \ge
\cdim$ and $\deg(\matcol{\mat{F}}{j}) < \degMod_j = \deg(\modulus_j)$, and let
$\degMod \ge \rdim$ such that $\degMod \ge \degMod_1 + \cdots + \degMod_\cdim$
and $\sumVec{\cdeg{\mat{P}}} \le \degMod$. Then $\mat{P} \mat{F} \bmod
\Modulus$ can be computed in $\softO{\rdim^{\expmatmul-1} \degMod}$ operations.
\end{lem}
\begin{proof}
\vspace{-0.15cm}
Using notation from Lemma~\ref{lem:modsys_parlin}, we let $\expand{\mat{P}} \in
\polMatSpace[\rdim][\expand{\rdim}]$ such that $\mat{P} = \expand{\mat{P}}
\expandMat$ and $\deg(\expand{\mat{P}}) < \lceil \sumVec{\cdeg{\mat{P}}} /
\rdim \rceil$. As above, $\expand{\sys} = \expandMat \sys \bmod \Modulus$ can
be computed in time $\softO{\rdim \degMod}$. Here we want to compute $\mat{P}
\sys \bmod \Modulus = \expand{\mat{P}} \, \expand{\sys} \bmod \Modulus$.

We have $\deg(\expand{\mat{P}}) \le \lceil \degMod / \rdim \rceil \le 2
\degMod/\rdim$. Since $\sumVec{\cdeg{\expand{\mat{F}}}} < \degMod$ and $\cdim
\le \rdim \le \expand{\rdim} \le 2\rdim$, $\expand{\mat{F}}$ can be partially
linearized into $\bigO{\rdim}$ columns of degree $\bigO{\degMod/\rdim}$. Then,
$\expand{\mat{P}} \, \expand{\sys}$ is computed in $\softO{\rdim^{\expmatmul-1}
\degMod}$ operations.
The $j$-th column of $\expand{\mat{P}} \, \expand{\sys}$ has $\expand{\rdim}
\le 2\rdim$ rows and degree less than $\degMod_j + 2\degMod/\rdim$: it can be
reduced modulo $\modulus_j$ in $\softO{\degMod + \rdim \degMod_j}$
operations~\cite[Chapter~9]{vzGathen13}; summing over $1 \le j \le \cdim$ with
$\cdim \le \rdim$, this is in $\softO{\rdim \degMod}$.
\end{proof}

\begin{proof}[of Theorem~\ref{thm:modsys}]
The correctness and the cost $\softO{\nbun^{\expmatmul-1} \degMod}$ for
Steps~\textbf{1} and \textbf{2} of Algorithm~\ref{algo:modsys} follow
from~\cite[Theorem~1.4]{JeNeScVi15} and Proposition~\ref{prop:modsys_one}. With
the costs of Steps \textbf{3.c} and~\textbf{3.e} given in
Proposition~\ref{prop:algo:knownmindeg_modsys} and Lemma~\ref{lem:residual}, we
obtain the announced cost bound.

Now, using notation in Step~\textbf{3}, suppose $\msb^{(1)}$ and $\msb^{(2)}$
are the $\shifts$- and $\rdeg[\shifts]{\msb^{(1)}}$-Popov solution bases for
$(\Modulus^{(1)},\sys^{(1)})$ and $(\Modulus^{(2)},\mat{R})$. Then
$\msb^{(2)}\msb^{(1)}$ is a solution basis for $(\Modulus,\sys)$: if
$\row{p}$ is a solution for $(\Modulus,\sys)$, it is one for
$(\Modulus^{(1)},\sys^{(1)})$ and thus $\row{p} = \boldsymbol{\lambda}
\msb^{(1)}$ for some $\boldsymbol{\lambda}$, and it is one for
$(\Modulus^{(2)},\sys^{(2)})$ so that $\row{p} \sys^{(2)} =
\boldsymbol{\lambda} \msb^{(1)} \sys^{(2)} = \boldsymbol{\lambda} \mat{R} =
\row{0} \bmod \Modulus^{(2)}$ and thus $\boldsymbol{\lambda} =
\boldsymbol{\mu} \msb^{(2)}$ for some $\boldsymbol{\mu}$; then $\row{p} =
\boldsymbol{\mu} \msb^{(2)}\msb^{(1)}$.

Then $\msb^{(2)}\msb^{(1)}$ is an $\shifts$-minimal solution basis for
$(\Modulus,\sys)$ and its $\shifts$-Popov form has $\shifts$-pivot degree
$\minDegs^{(1)} + \minDegs^{(2)}$~\cite[Section 3]{JeNeScVi16}. The correctness
follows from Proposition~\ref{prop:algo:knownmindeg_modsys}.
\end{proof}

\section{Fast computation of the shifted Popov form of a matrix}
\label{sec:popov}

\vspace{-0.15cm}
\subsection{Fast shifted Popov form algorithm}
\label{subsec:algo:popov}

Our fast method for computing the $\shifts$-Popov form of a nonsingular
$\mat{A} \in \polMatSpace$ uses two steps, as follows.
\vspace{-0.15cm}
\begin{enumerate}[{\bf 1.}]
      \setlength{\itemsep}{0cm}
  \item Compute the Smith form of $\mat{A}$, giving the moduli $\Modulus$,
    and a corresponding right unimodular transformation, giving the equations
    $\sys$, so that $\mat{A}$ is a solution basis for $(\Modulus,\sys)$.
  \item Find the $\shifts$-Popov solution basis for $(\Modulus,\sys)$.
\end{enumerate}
We first show the correctness of this approach.

\vspace{-0.1cm}
\begin{lem}
  \label{lem:reduction}
  Let $\mat{A} \in \polMatSpace[\rdim]$ be nonsingular and $\smith = \mat{U}
  \mat{A} \mat{V}$ be the Smith form of $\mat{A}$ with $\mat{U}$ and $\mat{V}$
  unimodular. Let $\Modulus \in \modSpace^\rdim$ and $\sys \in
  \polMatSpace[\rdim]$ be such that $\smith = \diag{\Modulus}$ and $\sys =
  \mat{V} \bmod \Modulus$. Then $\mat{A}$ is a solution basis for
  $(\Modulus,\sys)$.
\end{lem}
\begin{proof}
  Let $\row{p} \in \polMatSpace[1][\rdim]$. If $\row{p}$ is in the row space of
  $\mat{A}$ then $\row{p}$ is a solution for $(\Modulus,\sys)$ since $\mat{A}
  \mat{V} = \mat{U}^{-1} \smith$ with $\mat{U}^{-1}$ over $\polRing$. Now if
  $\row{p} \sys = 0 \bmod \Modulus$, then $\row{p} \mat{V} = \row{q} \smith$
  for some $\row{q}$ and $\row{p} = \row{q} \mat{U} \mat{A}$ is in the row
  space of $\mat{A}$.
\end{proof}
\vspace{-0.15cm}

Concerning the cost of Step~\textbf{1}, such $\Modulus$ and $\sys$ can be
obtained in expected $\softO{\rdim^\expmatmul \deg(\mat{A})}$ operations, by
computing
\begin{enumerate}[{\bf {1}.a}]
  \setlength{\itemsep}{0cm}
  \item $\mat{R}$ a row reduced form of $\mat{A}$ \cite[Theorem 18]{GuSaStVa12},
  \item $\diag{\Modulus}$ the Smith form of $\mat{R}$ \cite[Algorithm 12]{Storjohann03},
  \item $(\anyMat,\sys)$ a reduced Smith transform for $\mat{R}$ \cite[Figure 3.2]{Gupta11};
\end{enumerate}
as in~\cite[Figure 6.1]{Gupta11}, Steps~\textbf{1.b} and \textbf{1.c} should be
performed in conjunction with the preconditioning techniques detailed
in~\cite{KaKrSa90}. One may take for $\Modulus$ only the nontrivial Smith factors,
and for $\sys$ only the nonzero columns of the transform.

The product of the moduli in $\Modulus$ is $\det(\mat{A})$ so that the sum of
their degrees is $\deg(\det(\mat{A}))$. Then, according to
Theorem~\ref{thm:modsys}, Step~\textbf{2} of the algorithm outlined above costs
$\softO{\rdim^{\expmatmul-1} \deg(\det(\mat{A}))}$ operations. Thus this
algorithm solves Problem~\ref{pbm:popov} in expected $\softO{\rdim^\expmatmul
  \deg(\mat{A})}$ field operations.

\subsection{Reducing to almost uniform degrees}
\label{subsec:partial_linearization}

In this subsection, we use the partial linearization techniques
from~\cite[Section 6]{GuSaStVa12} to prove the following result.

\vspace{-0.2cm}
\begin{prop}
  \label{prop:parlin_reduction}
  Let $\mat{A} \in \polMatSpace[\rdim]$ be nonsingular and let $\shifts \in
  \shiftSpace$. With no field operation, one can build a nonsingular
  $\expand{\mat{A}} \in \polMatSpace[\expand{\rdim}]$ and a shift $\shifts[u]
  \in \shiftSpace[\expand{\rdim}]$ such that $\expand{\rdim} \le 3\rdim$,
  $\deg(\expand{\mat{A}}) \le \lceil \degDet / \rdim \rceil$, and the
  $\shifts$-Popov form of $\mat{A}$ is the principal $\rdim \times \rdim$
  submatrix of the $\shifts[u]$-Popov form of $\expand{\mat{A}}$.
\end{prop}
\vspace{-0.15cm}
With the algorithm in the previous subsection, this implies
Theorem~\ref{thm:popov}. In the specific case of Hermite form computation, for
which there is a deterministic algorithm with cost bound
$\softO{\rdim^\expmatmul \deg(\mat{A})}$ \cite{ZhoLab16}, one can verify that
this leads to a \emph{deterministic} algorithm using $\softO{\rdim^\expmatmul
\lceil \degDet / \rdim \rceil}$ operations. (However, for $\shifts=\unishift$
this does not give a $\softO{\rdim^\expmatmul \lceil \degDet / \rdim \rceil}$
\emph{deterministic} algorithm for the Popov form
using~\cite{GuSaStVa12,SarSto11}, since the corresponding $\shifts[u]$ is
$(\unishift,t,\ldots,t)$ with $t \ge \deg(\mat{A})$.)

\vspace{-0.15cm}
\begin{dfn}[Column partial linearization]
  \label{dfn:parlin}
  ~\\ Let $\mat{A} \in \polMatSpace[\rdim]$ and $\degLins = (\degLin_i)_i \in
  \NN^\rdim$. Then let $\degExp = 1 + \lfloor (\degLin_1 + \cdots +
  \degLin_\rdim) / \rdim \rfloor$, let $\quoExp_i \ge 1$ and $0 \le \remExp_i <
  \degExp$ be such that $\degLin_i = (\quoExp_i -1) \degExp + \remExp_i$ for $1
  \le i \le \rdim$, let $\expand{\rdim} = \quoExp_1 + \cdots + \quoExp_\rdim$,
  and let $\expandMat = \trsp{[\idMat[] | \trsp{\mat{E}}]} \in
  \polMatSpace[\expand{\rdim}][\rdim]$ be the expansion-compression matrix with
  $\idMat[]$ the identity matrix and
  \begin{equation}
    \label{eqn:expandMatbis}
    \mat{E} = 
    \left[\begin{smallmatrix}
      X^\degExp \\
      \vdots \\
      X^{(\quoExp_1-1)\degExp} \\
      & \ddots \\
      & & X^\degExp \\
      & & \vdots \\
      & & X^{(\quoExp_\nbun-1)\degExp}
    \end{smallmatrix}\right].
  \end{equation}
  The \emph{column partial linearization
    $\colParLin{\minDegs}{\mat{A}} \in \polMatSpace[\expand{\rdim}]$ of
    $\mat{A}$} is defined as follows:
    \vspace{-0.2cm}
  \begin{itemize}
  \setlength\itemsep{0cm}
    \item the first $\rdim$ rows of $\colParLin{\minDegs}{\mat{A}}$ form
      the unique matrix $\expand{\mat{A}} \in
      \polMatSpace[\rdim][\expand{\rdim}]$ such that $\mat{A} =
      \expand{\mat{A}} \expandMat$ and $\expand{\mat{A}}$ has all columns of
      degree less than $\degExp$ except possibly those at indices $\rdim +
      (\quoExp_1-1) + \cdots + (\quoExp_i-1)$ for $1\le i \le \rdim$,
    \item for $1\le i \le \rdim$, the row $\rdim + (\quoExp_1-1) + \cdots +
      (\quoExp_{i-1}-1) +1$ of $\colParLin{\minDegs}{\mat{A}}$ is $[ 0, \ldots,
      0, -X^\degExp, 0, \ldots, 0, 1, 0, \ldots, 0]$ where
      $-X^\degExp$ is at index $i$ and $1$ is on the diagonal,
    \item for $1\le i \le \rdim$ and $2 \le j \le \quoExp_i-1$, the row
      $\rdim + (\quoExp_1-1) + \cdots + (\quoExp_{i-1}-1) + j$ of
      $\colParLin{\minDegs}{\mat{A}}$ is $[ 0, \ldots, 0, -X^\degExp, 1, 0,
      \ldots, 0]$ where $1$ is on the diagonal.
  \end{itemize}
\end{dfn}

\vspace{-0.1cm}
Defining the \emph{row partial linearization} $\rowParLin{\minDegs}{\mat{A}}$
of $\mat{A}$ similarly, both linearizations are related by
$\rowParLin{\minDegs}{\mat{A}} = \trsp{\colParLin{\minDegs}{\trsp{\mat{A}}}}$.

Now we show that for a well-chosen $\shifts[u]$, one can directly read the
$\shifts$-Popov form of $\mat{A}$ as a submatrix of the $\shifts[u]$-Popov form
of $\rowParLin{\minDegs}{\mat{A}}$ (resp. $\colParLin{\minDegs}{\mat{A}}$).

\vspace{-0.1cm}
\begin{lem}
  \label{lem:parlin_reduction}
  Let $\mat{A} \in \polMatSpace[\rdim]$ be nonsingular, $\shifts
  \in~\shiftSpace$, $\popov$ be the $\shifts$-Popov form of $\mat{A}$, and
  $\minDegs \in \NN^\rdim$. We have that:
  \vspace{-0.15cm}
  \begin{enumerate}[(i)]
    \setlength\itemsep{0cm}
    \item if $\expand{\rdim}$ is the dimension of
      $\rowParLin{\minDegs}{\mat{A}}$ and $\shifts[u] = (\shifts,t,\ldots,t)$
      is in $\shiftSpace[\expand{\rdim}]$ with $t \ge
      \max(\shifts)+\deg(\popov)$, then the $\shifts[u]$-Popov form of
      $\rowParLin{\minDegs}{\mat{A}}$ is $\begin{bmatrix} \popov & \mathbf{0}
        \\ \anyMat & \idMat[] \end{bmatrix}$;
    \item if $\expand{\rdim}$ is the dimension of
      $\colParLin{\minDegs}{\mat{A}}$, $\mat{E}$ is as
      in~\eqref{eqn:expandMatbis}, and $\shifts[u] = (\shifts,\shifts[t]) \in
      \shiftSpace[\expand{\rdim}]$ for any $\shifts[t] \in
      \shiftSpace[\expand{\rdim}-\rdim]$, then the $\shifts[u]$-Popov form of
      $\colParLin{\minDegs}{\mat{A}}
        \begin{bmatrix}
          \idMat[] & \mat{0} \\
          \mat{E} & \idMat[]
        \end{bmatrix}$
      is
      $\begin{bmatrix}
        \popov & \mat{0} \\
        \mat{0} & \idMat[]
      \end{bmatrix}$;
    \item if $\expand{\rdim}$ is the dimension of
      $\colParLin{\minDegs}{\mat{A}}$ and $\shifts[u] = (\shifts,t,\ldots,t)$
      is in $\shiftSpace[\expand{\rdim}]$ with $t \ge
      \max(\shifts)+\deg(\popov)$, then the $\shifts[u]$-Popov form of
      $\colParLin{\minDegs}{\mat{A}}$ is $\begin{bmatrix} \popov & \mathbf{0}
        \\ \anyMat & \idMat[] \end{bmatrix}$.
  \end{enumerate}
\end{lem}
\begin{proof}
  $(i)$
  $\rowParLin{\minDegs}{\mat{A}}$ is left-unimodularly equivalent to 
  $[\begin{smallmatrix}
    \mat{A} & \mathbf{0} \\
    \mat{B} & \idMat[]
  \end{smallmatrix}]$
  for some $\mat{B} \in \polMatSpace[(\expand{\rdim}-\rdim)][\rdim]$
  \cite[Theorem 10 (i)]{GuSaStVa12}.
  Then, let $\mat{R}$ be the remainder of $\mat{B}$ modulo $\popov$, that is,
  the unique matrix in $\polMatSpace[(\expand{\rdim}-\rdim)][\rdim]$ which has
  column degree bounded by the column degree of $\popov$
  componentwise and such that $\mat{R} = \mat{B} + \mat{Q}\popov$ for some
  matrix $\mat{Q}$ (see for example~\cite[Theorem~6.3-15]{Kailath80}, noting
  that $\popov$ is $\unishift$-column reduced).

  Let $\mat{W}$ denote the unimodular matrix such that $\popov = \mat{W}
  \mat{A}$. Then, 
  $[\begin{smallmatrix}
    \mat{W} & \mathbf{0} \\
    \mat{Q}\mat{W}  & \idMat[]
  \end{smallmatrix}]
  [\begin{smallmatrix}
    \mat{A} & \mathbf{0} \\
    \mat{B} & \idMat[]
  \end{smallmatrix}]
  =
  [\begin{smallmatrix}
    \popov & \mathbf{0} \\
    \mat{R} & \idMat[]
  \end{smallmatrix}]$
  is left-unimodularly equivalent to $\rowParLin{\minDegs}{\mat{A}}$.
  Besides, since
  $\deg(\mat{R}) < \deg(\popov)$, we have that
  $[\begin{smallmatrix}
    \popov & \mathbf{0} \\
    \mat{R} & \idMat[]
  \end{smallmatrix}]$
  is in $\shifts[u]$-Popov form by choice of $t$.

  $(ii)$
  The matrix
  $[\begin{smallmatrix}
    \popov & \mat{0} \\
    \mat{0} & \idMat[]
  \end{smallmatrix}]$
  is obviously in $\shifts[u]$-Popov form: it remains to prove that it is
  left-unimodularly equivalent to 
  $\colParLin{\minDegs}{\mat{A}}
  [\begin{smallmatrix}
    \idMat[] & \mat{0} \\
    \mat{E} & \idMat[]
  \end{smallmatrix}]$.
  Let $\mat{T}$ denote the trailing principal submatrix $\mat{T} =
  \colParLin{\minDegs}{\mat{A}}_{\rdim+1 \dots \expand{\rdim},\rdim+1 \dots
  \expand{\rdim}}$, and let $\mat{W}$ be the unimodular matrix such that
  $\mat{W} \popov = \mat{A}$. Then, $\mat{T}$ is unit lower triangular, thus
  unimodular, and by construction of $\colParLin{\minDegs}{\mat{A}}$, for some
  matrix $\mat{B}$ we have $\colParLin{\minDegs}{\mat{A}}
  [\begin{smallmatrix}
    \idMat[] & \mat{0} \\
    \mat{E} & \idMat[]
  \end{smallmatrix}]
  =
  [\begin{smallmatrix}
    \mat{A} & \mat{B} \\
    \mat{0} & \mat{T}
  \end{smallmatrix}]
  =
  [\begin{smallmatrix}
    \mat{W} & \mat{B} \\
    \mat{0} & \mat{T}
  \end{smallmatrix}]
  [\begin{smallmatrix}
    \popov & \mat{0} \\
    \mat{0} & \idMat[]
  \end{smallmatrix}]
  $.

  $(iii)$
  From $(ii)$, $\colParLin{\minDegs}{\mat{A}}$ is left-unimodularly equivalent
  to
  $[\begin{smallmatrix}
    \popov & \mat{0} \\
    \mat{0} & \idMat[]
  \end{smallmatrix}]
  [\begin{smallmatrix}
    \idMat[] & \mat{0} \\
    -\mat{E} & \idMat[]
  \end{smallmatrix}]
  =
  [\begin{smallmatrix}
    \popov & \mat{0} \\
    -\mat{E} & \idMat[]
  \end{smallmatrix}]
  $. Using arguments in the proof of $(i)$ above, by choice of $t$ the
  $\shifts[u]$-Popov form of
  $[\begin{smallmatrix}
    \popov & \mat{0} \\
    -\mat{E} & \idMat[]
  \end{smallmatrix}]$
  is 
  $[\begin{smallmatrix}
    \popov & \mat{0} \\
    \mat{R} & \idMat[]
  \end{smallmatrix}]$
  with $\mat{R}$ the remainder of $-\mat{E}$ modulo $\popov$.
\end{proof}
\vspace{-0.15cm}
In the usual case where $\deg(\popov)$ is not known \emph{a priori}, one may
choose $t$ using the inequality $\deg(\popov) \le \deg(\det(\popov)) =
\deg(\det(\mat{A})) \le \rdim \deg(\mat{A})$. 

This result implies Proposition~\ref{prop:parlin_reduction}
thanks to the following remark from~\cite{GuSaStVa12}. Let $\pi_1,\pi_2$ be
permutation matrices such that $\mat{B} = \pi_1\mat{A}\pi_2 = [b_{i,j}]_{ij}$
satisfies
$\deg(b_{i,i}) \ge \deg(b_{j,k})$ for all $j,k\ge i$ and $1\le i\le \rdim$.
Defining $\shifts[d] = (d_i)_i \in \NN^\rdim$ by
$d_{i} = \overline{\deg}(b_{i,i}) = \left\{ \begin{array}{ll}
  \deg(b_{i,i}) & \text{ if } b_{i,i} \neq 0 \\
  0  & \text{ otherwise}
\end{array} \right.$,
we have $d_1 + \cdots + d_\rdim \le \degDet$ by definition of $\degDet$
in~\eqref{eqn:degDet}. Let $\minDegs = \pi_1^{-1} \shifts[d]$, where
$\shifts[d]$ is seen as a column vector, and $\boldsymbol\gamma =
\cdeg{\rowParLin{\minDegs}{\mat{A}}}$. Then the matrix $\expand{\mat{A}} =
\colParLin{\boldsymbol\gamma}{\rowParLin{\minDegs}{\mat{A}}}$ is
$\expand{\rdim} \times \expand{\rdim}$ with $\expand{\rdim} < 3\rdim$, and we
have $\deg(\expand{\mat{A}}) \le \lceil \degDet/\rdim \rceil$ \cite[Corollary
3]{GuSaStVa12}.  Lemma~\ref{lem:parlin_reduction} further shows that the
$\shifts$-Popov form of $\mat{A}$ is the principal $\rdim \times \rdim$
submatrix of the $\shifts[u]$-Popov form of $\expand{\mat{A}}$, for the shift
$\shifts[u] = (\shifts,t,\ldots,t) \in \shiftSpace[\expand{\rdim}]$ with $t =
\max(\shifts) + \rdim \deg(\mat{A})$.

\medskip \textbf{Acknowledgements.} The author sincerely thanks the anonymous
reviewers for their careful reading and detailed comments, which were very
helpful for preparing the final version of this paper. He also thanks C.-P.
Jeannerod, G. Labahn, \'E. Schost, A. Storjohann, and G. Villard for their
interesting and useful comments. The author gratefully acknowledges financial
support provided through the international mobility grants \emph{Explo'ra Doc
from R\'egion Rh\^one-Alpes}, \emph{PALSE}, and \emph{Mitacs Globalink -
Inria}.

\begin{tiny}

\end{tiny}

\end{document}